\documentclass[twoside]{article}

\usepackage{multirow}
\usepackage{caption, subfigure, lineno, wrapfig}
\newcommand{\SBM}{\mathbf{SBM}}
\usepackage{tikz}

\usepackage{a4,geometry}
\usepackage{enumitem}
\usepackage{graphicx}
\usepackage{amsmath,amssymb,amsthm,mathtools}
\usepackage{paralist}
\usepackage{longtable}
\usepackage{colortbl}
\usepackage{hhline}
\usepackage{bm}
\usepackage{xspace}
\usepackage{url}
\usepackage{fullpage, prettyref}
\usepackage{boxedminipage}
\usepackage{wrapfig}
\usepackage{subfigure}
\usepackage{ifthen}
\usepackage{color}
\usepackage{xcolor}
\usepackage{framed}
\usepackage{algorithm}
\usepackage{algpseudocode}
\usepackage[pagebackref,letterpaper=true,colorlinks=true,pdfpagemode=none,urlcolor=blue,linkcolor=blue,citecolor=violet,pdfstartview=FitH]{hyperref}
\usepackage{fullpage}
\usepackage{float}
\usepackage{rotating}

\usepackage{thmtools}
\usepackage{thm-restate}
\newtheorem{theorem}{Theorem}[section]

\newtheorem{lemma}[theorem]{Lemma}
\newtheorem{claim}[theorem]{Claim}

\newtheorem{definition}[theorem]{Definition}

\usepackage{thm-restate}

\newcommand{\ignore}[1]{}

\newcommand{\cA}{{\cal A}}

\newcommand{\eps}{\varepsilon}

\newcommand{\poly}{\mathrm{poly}}

\newcommand{\wt}{\widetilde}

\newcommand{\bone}{{\bf 1}}

\newcommand{\bX}{\boldsymbol{X}}

\newcommand{\EX}{\hbox{\bf E}}

\newcommand{\eqdef}{:=}

\newcommand{\wgt}{\mathrm{wt}}

\newcommand{\extract}{{\tt Extract}}
\newcommand{\decompose}{{\tt Decompose}}

\newcommand{\Sec}[1]{\hyperref[sec:#1]{\S\ref*{sec:#1}}} 
\newcommand{\Eqn}[1]{\hyperref[eq:#1]{(\ref*{eq:#1})}} 
\newcommand{\Fig}[1]{\hyperref[fig:#1]{Fig.\,\ref*{fig:#1}}} 
\newcommand{\Tab}[1]{\hyperref[tab:#1]{Tab.\,\ref*{tab:#1}}} 
\newcommand{\Thm}[1]{\hyperref[thm:#1]{Theorem\,\ref*{thm:#1}}} 
\newcommand{\Fact}[1]{\hyperref[fact:#1]{Fact\,\ref*{fact:#1}}} 
\newcommand{\Lem}[1]{\hyperref[lem:#1]{Lemma\,\ref*{lem:#1}}} 
\newcommand{\Prop}[1]{\hyperref[prop:#1]{Prop.~\ref*{prop:#1}}} 
\newcommand{\Cor}[1]{\hyperref[cor:#1]{Corollary~\ref*{cor:#1}}} 
\newcommand{\Conj}[1]{\hyperref[conj:#1]{Conjecture~\ref*{conj:#1}}} 
\newcommand{\Def}[1]{\hyperref[def:#1]{Definition~\ref*{def:#1}}} 
\newcommand{\Alg}[1]{\hyperref[alg:#1]{Alg.~\ref*{alg:#1}}} 
\newcommand{\Clm}[1]{\hyperref[clm:#1]{Claim~\ref*{clm:#1}}} 
\newcommand{\Obs}[1]{\hyperref[obs:#1]{Observation~\ref*{obs:#1}}} 
\newcommand{\Rem}[1]{\hyperref[rem:#1]{Remark~\ref*{rem:#1}}} 
\newcommand{\Con}[1]{\hyperref[con:#1]{Construction~\ref*{con:#1}}} 
\newcommand{\Step}[1]{\hyperref[step:#1]{Step~\ref*{step:#1}}} 
\newcommand{\Assumption}[1]{\hyperref[assm:#1]{Assumption\,\ref*{assm:#1}}} 

\makeatletter

\newcommand{\ProbabilityRender}[2]{
  \@ifnextchar\bgroup%
  {\renderwithdist{#1}{#2}}
   {\singlervrender{#1}{#2}}
}
\newcommand{\singlervrender}[2]{%
   \ensuremath{\mathchoice
       {{#1}\left[ #2 \right]}
       {{#1}[ #2 ]}
       {{#1}[ #2 ]}
       {{#1}[ #2 ]}
   }
}
\newcommand{\renderwithdist}[3]{%
   \@ifnextchar\bgroup
   {\superfancyrender{#1}{#2}{#3}}
   {\ensuremath{\mathchoice
      {\underset{#2}{#1}\left[ #3 \right]}
      {{#1}_{#2}[ #3 ]}
      {{#1}_{#2}[ #3 ]}
      {{#1}_{#2}[ #3 ]}
     }
   }
}
\newcommand{\superfancyrender}[5]{
   \ensuremath{\mathchoice
      {\underset{#1}{{#1}}\left#4 #3 \right#5}
      {{#1}_{#2}#4 #3 #5}
      {{#1}_{#2}#4 #3 #5}
      {{#1}_{#2}#4 #3 #5}
   }
}

\newcommand{\stc}{\tau}
\newcommand{\stch}{\eps}

\begin{document}
	\title{\Large Spectral Triadic Decompositions of Real-World Networks}
	\author{Sabyasachi Basu\thanks{Department of Computer Science and Engineering, University of California, Santa Cruz {\href{mailto:sbasu3@ucsc.edu}{sbasu3@ucsc.edu}}}
		\and Suman Kalyan Bera\thanks{Apple {\href{mailto:sumankalyanbera@gmail.com}{sumankalyanbera@gmail.com}}}
		\and C. Seshadhri\thanks{Department of Computer Science and Engineering, University of California, Santa Cruz {\href{mailto:sesh@ucsc.edu}{sesh@ucsc.edu}}
			\newline
			{SB and CS are supported by NSF DMS-2023495 and CCF-1839317.}}}
	
	\date{}
	\maketitle
	\begin{abstract}
		A fundamental problem in mathematics and network analysis
		is to find conditions under which a graph can be partitioned 
		into smaller pieces. A ubiquitous tool for this partitioning is the
		Fiedler vector or discrete Cheeger inequality. These results relate the graph
		spectrum (eigenvalues of the normalized adjacency matrix) to the ability to break a graph into two
		pieces, with few edge deletions. An entire subfield of mathematics, called spectral
		graph theory, has emerged from these results. Yet these results do not say anything
		about the rich community structure exhibited by real-world networks, which typically
		have a significant fraction of edges contained in numerous densely clustered blocks. Inspired by
		the properties of real-world networks, we discover a new spectral condition that
		relates eigenvalue powers to a network decomposition into densely clustered blocks.
		We call this the \emph{spectral triadic decomposition}.
		Our relationship exactly predicts the existence of community structure, as commonly seen
		in real networked data. Our proof provides an efficient algorithm to produce
		the spectral triadic decomposition. We observe on numerous social, coauthorship, and citation network datasets that these decompositions have significant correlation with semantically meaningful communities.
	\end{abstract}
	\section{Introduction} \label{sec:intro}

	The existence of clusters or community structure is one of the most fundamental 
	properties of real-world networks. Across various scientific disciplines, be it biology,
	social sciences, or physics, the modern study of networks has often dealt with
	the community structure of these data. Procedures that discover community
	structure have formed an integral part of network science algorithmics. Despite
	the large variety of formal definitions of a community in a network, there
	is broad agreement that it constitutes a dense substructure in an overall sparse network. 
	Indeed, the discovery of local density (also called \emph{clustering coefficients}) goes
	back to the birth of network science.
	
	Even beyond network science, graph partitioning is a central problem in applied mathematics
	and the theory of algorithms. Determining when such a partitioning is possible is a fundamental
	question that straddles graph theory, harmonic analysis, differential geometry, and theoretical
	computer science. There is a large body of mathematical and scientific research on how to break
	up a graph into smaller pieces.
	
	An important mathematical tool for graph partitioning is 
	the \emph{discrete Cheeger inequality} or the \emph{Fiedler vector}. This result
	is the cornerstone of spectral graph theory and relates the eigenvalues
	of the graph Laplacian to the combinatorial structure. Consider an undirected
	graph $G = (V,E)$ with $n$ vertices. Let $d_i$ denote the degree of the vertex $i$.
	The \emph{normalized adjacency matrix}, denoted $\cA$,
	is the $n \times n$ matrix where the entry $\cA_{ij}$ is $1/\sqrt{d_i d_j}$
	if $(i,j)$ is an zero, and zero otherwise. (All diagonal entries are zero.)
	One can think of this entry as the ``weight" of the edge between $i$ and $j$.
	
	Let $\lambda_1 \geq \lambda_2 \ldots \geq \lambda_n$ denote the $n$ eigenvalues
	of the non-negative symmetric matrix $\cA$. The largest eigenvalue $\lambda_1$ is always one.
	A basic fact is that $\lambda_2 = 1$ iff $G$ is disconnected. 
	The discrete Cheeger inequality proves that if $\lambda_2$ is close to $1$ (has value $\geq 1-\eps$),
	then $G$ is ``close" to being disconnected. Formally, there exists a set $S$ of vertices
	that can be disconnected (from the rest of $G$) by removing an $O(\sqrt{\eps})$-fraction
    of edges incident to $S$. The cuts are measured by conductance, which is (roughly) the fraction of edges leaving a set. A low conductance
    cut can thus be separated by removing few edges.
	We can summarize these observations as:
	
	\medskip
	\noindent
	{\bf Basic fact:} \ \ Spectral gap is zero \ \ $\Longrightarrow$ \ \ $G$ is disconnected \\
	{\bf Cheeger bound:} \ \ Spectral gap is close to zero \ \ $\Longrightarrow$ \ \ $G$ can be disconnected by low conductance set
	
	\medskip
	
	There is a rich literature of generalizing this bound for higher-order
	networks and simplicial complices. We note that many modern algorithms for finding communities
	in real-world networks are based on the Cheeger inequality in some form. The seminal Personalized
	PageRank algorithm provides a local version of the Cheeger bound~\cite{ACL06}.
	
	For modern network analysis and community structure, there are several unsatisfying aspects of
	the Cheeger inequality. 
	Despite the variety of formal definitions of a community in a network, there
	is broad agreement that it constitutes many densely clustered substructures in an overall sparse network.
	The Cheeger inequality only talks of disconnecting $G$ into two parts. Even currently known generalizations
	of the Cheeger inequality only work for a constant number of parts~\cite{LeGh+14}. Real-world networks
	decompose into an extremely large of number of blocks/communities, and this number
	often scales with the network size~\cite{LeLaDa08,SeKoPi12}. Secondly, the Cheeger bound works when the spectral gap 
	is close to zero, which is often not true for real-world networks~\cite{LeLaDa08}.
	Real-world networks possess the small-world property~\cite{Kl00}. But this property
	implies large spectral gap. 
	Thirdly, Cheeger-type inequalities make no assertion on the interior of parts obtained. In community structure,
	we typically expect the interior to be dense and potentially assortative (possessing vertices
	of similar degree).
	
	\medskip
	
	The main question that we address: \emph{is there a spectral quantity that
		predicts the existence of real-world community structure?}
	
	\begin{figure*}
		\centering
		\begin{minipage}{.38\linewidth}
			\begin{subfigure}
				\centering
				{\includegraphics[width=\textwidth]{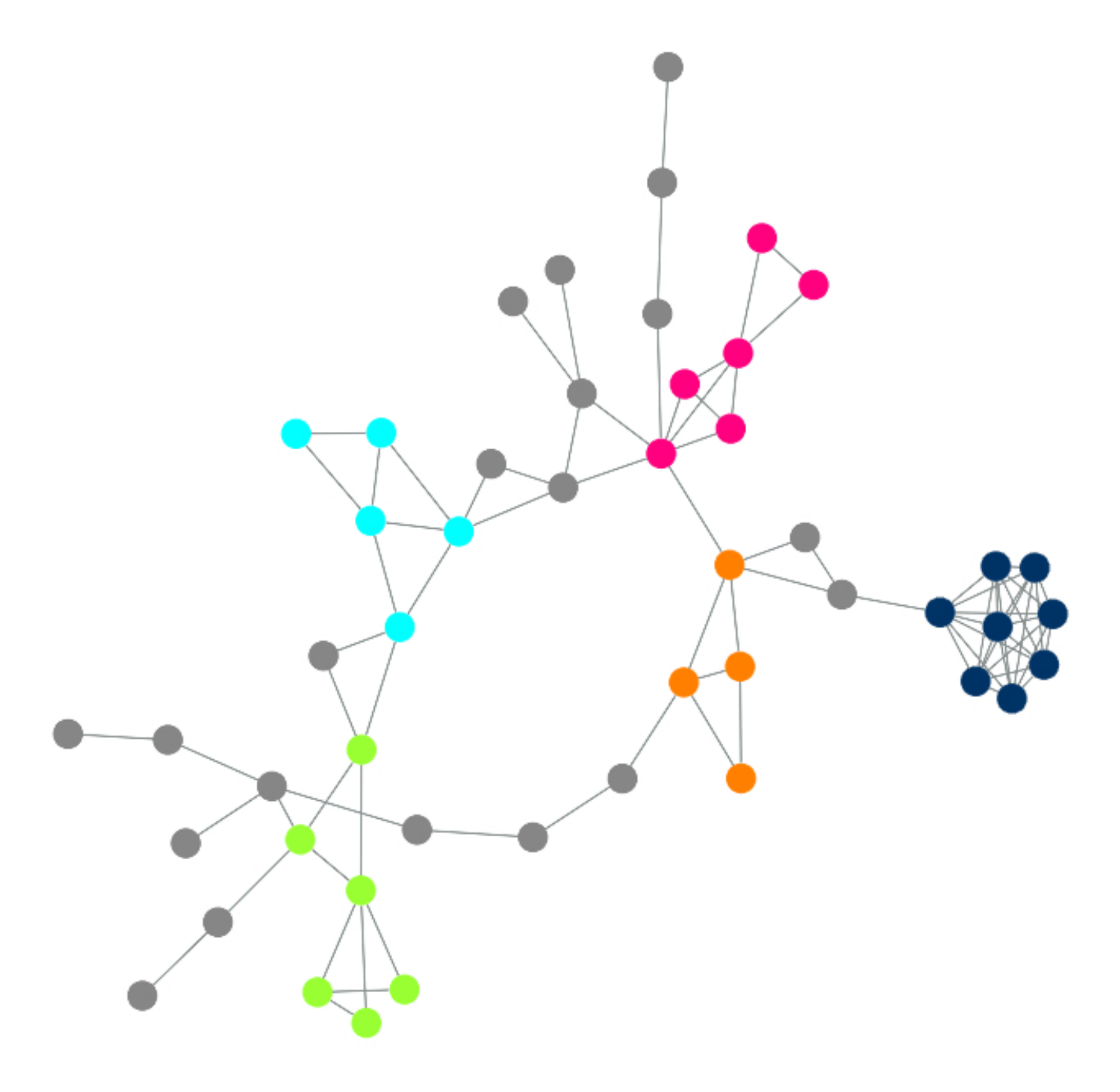}}
				\caption*{A toy decomposition: a subset of a co-authorship network.} 
			\end{subfigure}
		\end{minipage}
		\vline
		\begin{minipage}{.38\linewidth}
			\centering
			\begin{subfigure}
				\centering
				{\includegraphics[width=0.5\textwidth]{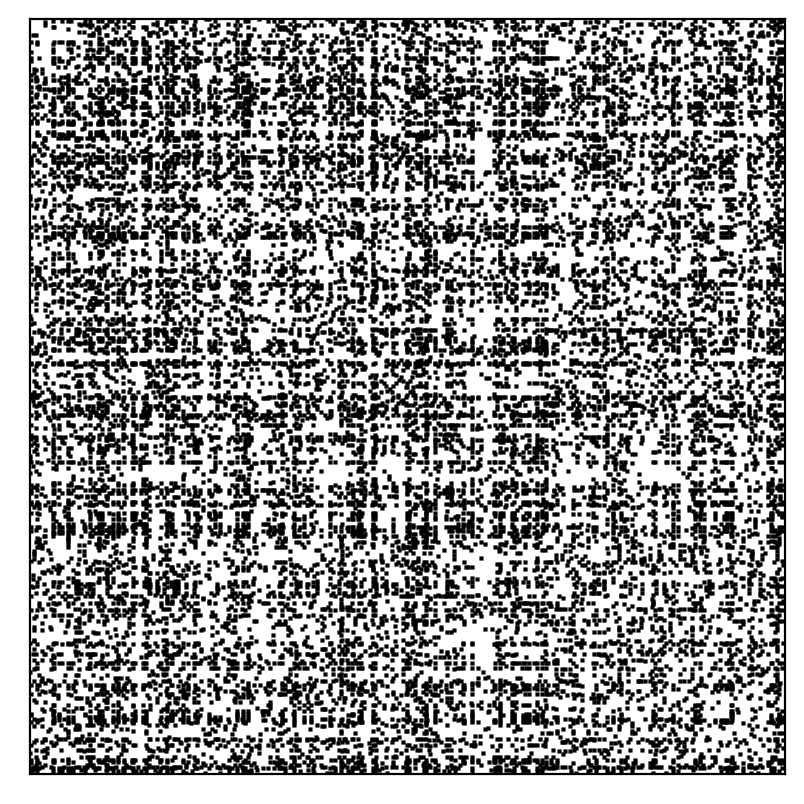}}
				\caption*{A subgraph from a facebook network}
			\end{subfigure} 
			\begin{subfigure}
				\centering
				{\includegraphics[width=0.5\textwidth]{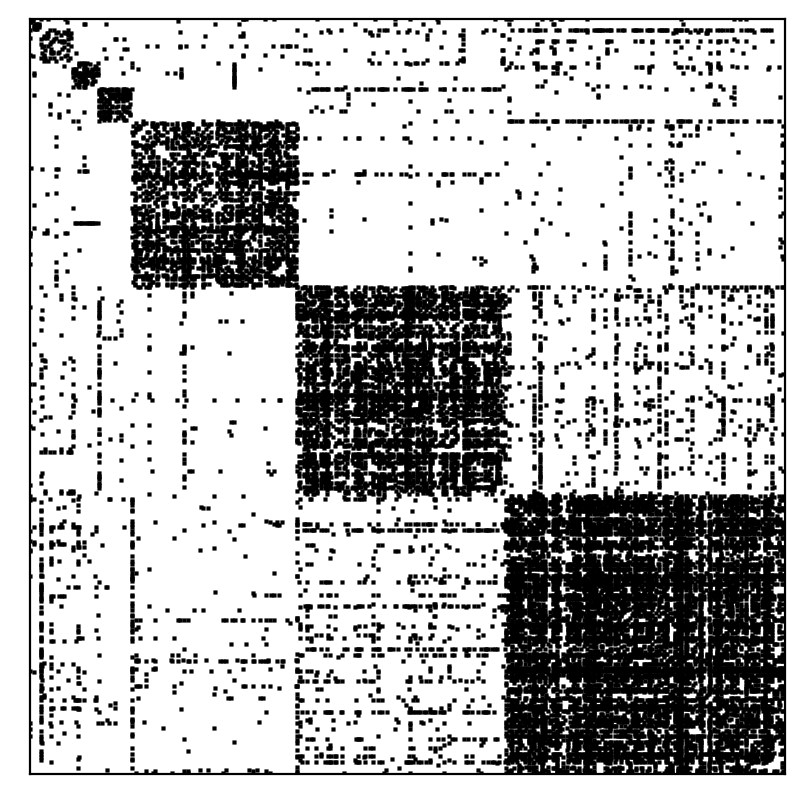}}
				\caption*{Subgraph ordered by extracted clusters}
			\end{subfigure} 
		\end{minipage}
		\caption{On the left, as a small example, we consider a subgraph induced by $155$ vertices and $\tau=0.49$ from a coauthorship network of Condensed Matter Physics researchers\cite{NewmanCondMat99}, and show a spectral triadic decomposition of the largest connected component, which has $49$ vertices.
			Each cluster is colored differently. We see how each cluster forms a densely connected component within an otherwise sparse graph. 
			Also note that the clusters vary in size. The gray vertices do not participate in the decomposition, since they do not
			add significant to the cluster structure. On the right, we look at the adjacency matrices pre and post decomposition. The top figure is a spy plot of the adjacency matrix of 488 connected vertices from a Facebook network (\cite{traud2012social},\cite{Traud:2011fs}) taken from the network repository\cite{nr}, a graph with $\tau=0.122$. As a demonstration,
			we compute the spectral triadic decomposition of this subnetwork. We group the columns/rows by the clusters in the spy plot on the bottom. The latent community structure
			is immediately visible. Note that there exists many such blocks of varying sizes.} \label{fig:spytoy}
	\end{figure*}
		\subsection{Main result} \label{sec:result}
	
	We take inspiration from a central property of real-world graphs, the abundance of triangles~\cite{WaSt98,SeKoPi12}.
	This abundance is widely seen across graphs that come by disparate domains. Recent work in network
	science and data mining have used the triangles to effectively cluster graphs. There is much evidence
	that the triangle structure aids finding communities in graphs~\cite{SaSePi+15,Ts15,BeGlLe16,TPM17}.
	
	In network science, the triangle count is often expressed in terms of the \emph{transitivity}
	or global clustering coefficient~\cite{Faust,WF94}.
	We define the \emph{spectral transitivity} of the graph $G$.
	
	\begin{definition} \label{def:spec-trans} The spectral transitivity of $G$, denoted $\stc(G)$,
		is defined as follows\footnote{If $G$ (or the normalized adjacency matrix $\cA$) are obvious from context, we
			simply refer to $\stc$ instead of $\stc(G)$.}. (Recall that the $\lambda_i$s are the eigenvalues of the normalized adjacency matrix.)
		\begin{equation}
			\stc(G) = \dfrac{\sum_{i\leq n} \lambda_i^3}{\sum_{i\leq n}\lambda_i^2}.
		\end{equation}
	\end{definition} 
	
	Standard arguments show
	that the spectral transitivity is a degree weighted transitivity. The numerator
	is a weighted sum over all triangles, while the denominator (squared Frobenius norm)
	is a weighted sum over edges (\Lem{stc}). 
	
	Observe that since $\lambda_i \leq 1$, $\stc \leq 1$. When $\stc$ reaches
	its maximum value of $1-1/(n-1)$, one can show that $G$ is a clique (\Lem{ratio-basic}).
	We formalize the notion of "clique-like" submatrices through the concept of uniformity.
	For a symmetric matrix $M$ and a subset $S$ of its columns/rows, we use $M|_S$ to denote
	the square submatrix restricted to $S$ (on both columns and rows).
	
    \begin{restatable}{definition}{uniform} \label{def:uniform} Let $\alpha \in (0,1]$. Let $\cA$ be the normalized adjacency matrix of a graph $G$. For any subset of vertices $S$, $\cA|_S$ is called $\alpha$-uniform if at least an $\alpha$-fraction of non diagonal entries have values at least $\alpha/(|S|-1)$.
		
		For $s \in S$, let $N(s,S)$ denote the neighborhood of $s$ in $S$ (we define
		edges by non-zero entries). An $\alpha$-uniform matrix is \emph{strongly $\alpha$-uniform}
		if for at least an $\alpha$-fraction of $s \in S$, $\cA|_{N(s,S)}$ is also $\alpha$-uniform.
	\end{restatable}
	
	Observe that the normalized adjacency matrix of a clique is (strongly) $1$-uniform. But submatrices
	of this matrix are not. Roughly speaking, a constant uniform submatrix corresponds to a dense
	subgraph of (say) size $k$ where the \emph{total} degrees (i.e. degrees in the original graph) of vertices is $\Theta(k)$.
	Strong uniformity is closely related to \emph{clustering coefficients}, which is the
	edge density of neighborhoods. It is well-known that real-world graphs have high clustering 
	coefficients~\cite{WaSt98,SeKoPi12}. A strongly uniform submatrix essentially exhibits high clustering coefficients.
	
	Our main theorem states that any graph with constant spectral transitivity
	can be decomposed into constant uniform blocks. We use $\|M\|_F$ to denote
	the Frobenius norm of matrix $M$.
	
	\begin{theorem}[Spectral Theorem] \label{thm:main-decomp} 
		Let $\cA$ be the normalized adjacency matrix of a graph with spectral transitivity $\stc$. 
		
		There exists a collection of disjoint sets of vertices $X_1, X_2, \ldots, X_k$
		satisfying the following conditions:
		\begin{compactenum}
        \item (Cluster structure) For all $i \leq k$, $\cA|_{X_i}$ is strongly $\poly(\stc)$-uniform.
        \item (Coverage) $\sum_{i \leq k} \|\cA|_{X_i}\|^2_F \geq \poly(\stc)\|\cA\|^2_F$.
		\end{compactenum}
        (The notation $\poly(\stc)$ denotes some fixed polynomial of $\stc$.)
	\end{theorem}
	
	We call this output the \emph{spectral triadic decomposition}. 
	Our proof also yields an efficient algorithm that computes the decomposition, whose running time
	is dominated by a triangle enumeration. Details in are given in \Thm{runtime} and \Sec{impl}.
	We thus answer the question posed at the end of the preceding section positively:

	\medskip
	\noindent
	{\bf Our result:} \ \ Spectral transitivity is high \ \ $\Longrightarrow$ \ \ $G$ decomposes into disjoint, dense clusters.
	
	\medskip
	\subsection{Significance of \Thm{main-decomp}} \label{sec:signif}
	
	One can think of \Thm{main-decomp} as a type of Cheeger inequality that is relevant to
	the structure of real-world social networks. We explain how it captures many of the 
	salient properties of clusters in real-world networks. In this discussion,
	we will assume that $\stc$ is a constant.
	
	{\bf The spectral transitivity:} We find it remarkable that
	a bound on a single spectral quantity, $\stc$, implies such a rich decomposition.
	The spectral transitivity $\stc$ captures a key property of real-world
	graphs, the abundance of triangles. While there is a rich body of empirical
	work on using triangles to cluster graphs, there is no theory explaining \emph{why}
	triangles are so useful. \Thm{main-decomp} gives a spectral-theoretic explanation.
	
	The spectral transitivity is a weighted version of the transitivity, which
	is typically around $0.1$ for real-world graphs. We also note that the final
	algorithm that computes the decomposition focuses on triangle cuts, which
	is a popular empirical technique for finding clusters in social networks~\cite{BeGlLe16,TPM17}.
	
	{\bf The strong uniformity of clusters:} Each cluster $X_i$ of the spectral triadic decomposition is (constant)
	strongly uniform. While there is no one definition of a``community" in real-world graphs,
	the definition of strong uniformity captures many basic concepts. Most importantly, $X_i$
	is internally dense in edges. Let $|X_i| = c$. Then $\Omega(c^2)$ entries in $X_i$
	are $\Omega(1/c)$, which (by averaging) implies that a constant fraction of $X_i$ involves vertices
	of degree $\Theta(c)$. Thus, a constant fraction of $X_i$ vertices have a constant fraction
	of their neighbors in $X_i$. Moreover, the submatrix of every neighborhood in $X_i$ is
	also uniform. This is quite consistent with the typical notion of a social network
	community. 
	
	Crucially, \Thm{main-decomp} gives a condition on the \emph{internal} structure of the decomposition.
	This addresses a key weakness of the Cheeger inequality.
	
	{\bf The coverage condition:} It is natural to measure the ``mass" of a matrix by the squared Frobenius norm.
	The clusters of spectral triadic decomposition of \Thm{main-decomp} capture a constant fraction
	of this squared norm. This is consistent with the fact that a constant fraction 
	of the edges in a real-world graph are \emph{not} community edges~\cite{Mi67,Gr83,Kl00,SeKoPi12}.
	Any decomposition into communities would avoid these "long-range" edges, excluding
	a constant fraction of the matrix mass. 
	
	{\bf Robustness to noise:} Taking the above point further, the non-community edges
	are often modeled as stochastic (or noisy). The underlying cluster structure of a real-world
	graph is robust to such perturbations. Adding (say) an Erd\H{o}s-R\'{e}nyi graph
	with $\Theta(n)$ edges to such a graph can only affect the spectral transitivity by a constant factor
	(by changing the Frobenius norm). \Thm{main-decomp} would only be affected by constant factors.
	Note that the spectral gap, on the other hand, can dramatically
	increase by such noise. 
	
	{\bf Spectral graph theory inspired by real-world graphs:} We consider \Thm{main-decomp}
	as opening up a new direction in spectral graph theory. At a mathematical level,
	\Thm{main-decomp} is like a Cheeger inequality, where a spectral condition implies
	a graph theoretic property. But all aspects of \Thm{main-decomp} (the notion of spectral
	transitivity and the properties of the decomposition) are inspired by the observed
properties of real-world graphs.

    \subsection{Perspectives on Community Detection} \label{sec:perspectives} 

    We list some connections between \Thm{main-decomp} and common issues in community detection.
    Note that \Thm{main-decomp} has no statistical
    assumption and does not reference any ground truth structure. 

   {\bf The number of communities:} A common challenge in many community detection methods is setting the number of communities,
    which is typically unknown ~\cite{FORTUNATO201075}. Moreover, most real-world networks have an extremely large number of communities  ~\cite{LLDM08, NR16}.
    A benefit of \Thm{main-decomp} is that $k$, the number of clusters, is not a parameter. Our algorithm
    that computes the decomposition has a single density parameter which is easy to set. Mathematically, it can be set 
    to $\stc$, a fixed function of the graph.

    {\bf The resolution limit:} Modularity is a classic objective used for community detection and the state-of-the-art Louvain
    method is based on optimizing this measure~\cite{New06}. Unfortunately, it is known to suffer from the ``resolution limit",
    and can sometimes miss communities below a size threshold~\cite{FB07}. We may be able to use \Thm{main-decomp} to derive
    explicit objectives that can avoid this limit.  \Thm{main-decomp} is not sensitive
    to the number of clusters/communities (as mentioned above) and only depends on a single global measure.
    We believe this could become one of the most important applications of \Thm{main-decomp}. 

    {\bf A statistical view:} \Thm{main-decomp} is a distribution-free statement, since it only involves
    a deterministic quantity, $\stc$. For any distribution that creates graphs with a constant $\tau$,
    \Thm{main-decomp} would be valid. It would be promising future work to see if one can use \Thm{main-decomp}
    to prove statistical conditions for feasibility of community detection.

    {\bf Requiring constant $\stc$:} \Thm{main-decomp} only gives non-trivial results when the $\stc$ is constant
    with respect to graph size, since there are $\poly(\stc)$ dependencies in all the bounds. Thus, it only
    deals triangle-rich settings, which does cover most social network applications. Indeed, in all our experiments,
    the real-world graphs have a large enough $\stc$ for getting results. We do note that there are numerous community detection
    settings where triangles might be too few, for example sparse Stochastic Block Model (SBM) settings, or sparse 
    planted partition models~\cite{CSX12}.

 \subsection{Empirical Evaluation} \label{sec:empirical}

We implement an algorithm that computes a spectral triadic decomposition of \Thm{main-decomp}. We
evaluate this algorithm on a variety of real-world datasets.
For context, we also compute
decompositions using a number of classic community detection/graph clustering methods: 
the Louvain algorithm~\cite{Louvain}, Infomap~\cite{infomap}, Label propagation~\cite{LabelProp}, and
$k$-way Spectral cuts~\cite{kway}. For a deeper comparison between methods, we also 
evaluated the algorithms on simple Stochastic Block Models (SBMs).
We focus on settings that create a large number of small, dense components. While this is
not the common setting, it is used in popular real-world network models like BTER~\cite{SeKoPi12}.

The details are given in \Sec{emp} and \Sec{other}.
We summarize our findings in \Fig{SBM1} and \Tab{noisesbmtab}, and in \Tab{comp-largest}. 

\begin{asparaitem}
    \item For all real-world datasets, the Spectral Triadic Decomposition outputs a large collection of dense clusters
        that cover a large fraction of the Frobenius norm and vertex set. For the simple SBM settings we experimented with,
        the Spectral Triadic Decomposition gets perfect recovery of the ground truth.
    \item In cases where we have vertex names (like coauthorship and citation networks), the spectral triadic clusters
        are semantically meaningful. We extract clusters of scientists in a subfield, or a collection of papers on a specific topic.
    \item The Infomap, Label Propagation, and $k$-way spectral clustering create extremely large, sparse clusters
        in these datasets. These clusters are not like communities, and often involve 50\% of all vertices. In SBM
         settings, the clusters are highly erroneous.
    \item The Louvain algorithm gives better results than the other algorithms, but also suffers from the problem
        of creating a few large clusters with low density. A typical cluster of Louvain is somewhat larger
        and less dense than a spectral triadic cluster. For SBMs, the Louvain algorithm make significant errors,
        but not as egregious as the other procedures.
    \item The semantically meaningful spectral triadic clusters are not discovered by Louvain (refer to \Sec{realsec}).
\end{asparaitem}
	
	\section{Related Work} \label{sec:related}
	Spectral graph theory is a deep field of study with much advancement over the past two decades. We refer the readers to the classic textbook by Chung \cite{Chung:1997}, and the tutorial \cite{SpielmanTutorial} and lecture notes \cite{SpielmanSAGT} by Spielman. 
	
	The cluster structure of real-world networks has attracted attention from the early days of network science~\cite{GiNe02,Newman03}.
	Fortunato's (somewhat dated) survey on community detection has details of the key results~\cite{FORTUNATO201075}. 
	There is no definitive model for social networks, but it is generally accepted that they have many dense clusters
	with sparse connections between them~\cite{CF06,LeLaDa08,SeKoPi12}. The study of triangles and neighborhood density
	goes back to the early days of social science theory~\cite{HoLe70,HL76,Burt04,Faust}. Early network science papers popularized the notion of clustering coefficients and transitivity as
	useful measures~\cite{WaSt98}. Some early work also use clustering coefficients to remove edges, but with the aim 
    of optimizing between metrics~\cite{Radicchi}.
	The use of triangles to find such clusters is a more recent development in network science. A number of contemporary
	results explicit use triangle information for algorithmic purposes~\cite{SaSePi+15,Ts15,BeGlLe16,TPM17}. Many of these results use
    triangle counts on edges as weights for either removal or insertion into communities.  Our main theorem is inspired by these applications. 
	
	While the Cheeger inequality by itself is not useful for real-world graph clustering, local versions
	of spectral clustering are extremely useful~\cite{ST08,AnChLa06}. We stress that the local versions do not
	relate the graph spectrum to the partitions but the algorithms bear striking similarities to the sweep cut procedure used to prove the Cheeger
	inequality. The approach is also central to the study of mixing times~\cite{LS88, JS89}. Alternatively, many results on the cluster structure of real-world graphs~\cite{LeLaDa08,GlSe12} use the Personalized PageRank method~\cite{AnChLa06}.
	Only a handful of local partitioning methods yield bounds on the internal structure of clusters~\cite{LeGh+14, kway, PSZ15}.
	
	Most relevant to our work is the result of Gupta, Roughgarden, and Seshadhri~\cite{GRS}. They prove a decomposition theorem
	for triangle-rich graphs, as measured by graph transitivity. Their main result shows that a triangle-dense graph
	can be clustered into dense clusters. The results of~\cite{GRS} do not have any spectral connection, nor 
	do they provide the kind of uniformity or coverage bounds of \Thm{main-decomp}. 
	Our main insight is in generalizations of their proof technique, which
	leads to connections with graph spectrum. We adapt the proof from ~\cite{GRS} to deal with normalized
	adjacency matrix, which adds many complications because of the non-uniformity of entries.

	\section{Preliminaries} \label{sec:prelims}
	
	We use $V, E, T$ to denote the sets of vertices, edges, and triangles of $G$, respectively.
	For any subgraph $H$ of $G$, we use $V_H, E_H, T_H$ to denote the corresponding
	sets within $H$. For any edge $e$, let $T_H(e)$ denote the 
	set of triangles in $H$ containing $e$.
	
	For any vertex $v$, let $d_v$ denote the degree of $v$ (in $G$). We stick to this notation for the rest of the article; unless specified otherwise, $d_v$ is never used to denote the degree in any induced subgraphs.
	
	We define \emph{weights} for edges and triangles. We will think of 
	edges and triangles as unordered sets of vertices. 
	
	\begin{definition} \label{def:weight} For any edge $e = (u,v)$, define the weight $\wgt(e)$
		to be $\frac{1}{d_u d_v }$. For any triangle $t = (u,v,w)$, define the weight $\wgt(t)$
		to be $\frac{1}{d_u d_v d_w}$.
		
		For any set $S$ consisting solely of edges or triangles, define $\wgt(S) = \sum_{s \in S} \wgt(s)$.
	\end{definition}

	Let $S \subseteq V$ be a subset of vertices, and let $\cA|_S$ denote the submatrix
	of $\cA$ restricted to $S$. We use $\lambda_i(S)$ to denote the $i$th largest
	eigenvalue of the symmetric submatrix $\cA|_S$. Abusing notation, we use $E_S$ and $T_S$
	to denote the edges and triangles contained in the graph induced on $S$.

	We state some standard facts that relate the sum of weights to sum of eigenvalue powers.

    \begin{claim} \label{clm:frob} $\sum_{i \leq |S|} \lambda^2_i(S) = 2 \sum_{e \in E(S)} \wgt(e)$
	\end{claim}
	
	\begin{proof} By the properties of the Frobenius norm of matrices, $\sum_{i \leq |S|} \lambda^2_i = \sum_{s,t \in S} \cA^2_{st}.$
		Note that $\cA_{st} = A_{st}/\sqrt{d_sd_t}$. Hence, $\sum_{s,t} \cA^2_{s,t} = 2 \sum_{e = (u,v) \in E(S)} 1/d_ud_v$.
		(We get a $2$-factor because each edge $(u,v)$ appears twice in the adjacency matrix.)
	\end{proof}
	
	\begin{claim} \label{clm:tri-wgt} $\sum_{i \leq |S|} \lambda^3_i(S) = 6\sum_{t \in T(S)} \wgt(t)$.
	\end{claim}
	
	\begin{proof} Note that $\sum_{i \leq |S|} \lambda^3_i(S)$ is the trace of $(\cA|_S)^3$. The diagonal entry
		$(\cA|_S)^3_{ii}$ is precisely $\sum_{s \in S} \sum_{s' \in S} \cA_{is} \cA_{ss'} \cA_{s'i}$. 
		Note that $\cA_{is} \cA_{ss'} \cA_{s'i}$ is non-zero iff $(i,s,s')$ form a triangle. In that case,
		$\cA_{is} \cA_{ss'} \cA_{s'i} = 1/\sqrt{d_id_s} \cdot 1/\sqrt{d_sd_{s'}} \cdot 1/\sqrt{d_{s'}d_i} = \wgt((i,s,s'))$.
		We conclude that $(\cA|_S)^3_{ii}$ is $2 \sum_{t \in T(S), t \ni i} \wgt(t)$.
		(There is a $2$ factor because every triangle is counted twice.)
		
		Thus, $\sum_{i \leq n} \lambda^3_i(S) = \sum_i 2 \sum_{t \in T, t \ni i} \wgt(t) = 2 \sum_{t \in T} \sum_{i \in t} \wgt(t)
		= 6\sum_{t \in T} \wgt(t)$. (The final $3$ factor appears because a triangle contains exactly $3$ vertices.)
	\end{proof}
	
	\begin{claim} \label{clm:tri-frob} $\sum_{t \in T(S)} \wgt(t) \leq \|\cA|_S\|^2_F/6$.
	\end{claim}
	
	\begin{proof} By \Clm{tri-wgt} $\sum_{t \in T(S)} \wgt(t) = \sum_{i \leq |S|}\lambda^3_i(S)/6$. The maximum eigenvalue
		of $\cA$ is $1$, and since $\cA|_S$ is a submatrix, $\lambda_1(S) \leq 1$ (Cauchy's interlacing theorem). 
		Thus, $\sum_{i \leq |S|} \lambda^3_i(S) \leq \sum_{i \leq |S|}\lambda^2_i(S) = \|\cA|_S\|^2_F$. 
	\end{proof}
	
	As a direct consequence of the previous claims applied on $\cA$, we get the following characterization of the spectral triadic content
	in terms of the weights.
	
	\begin{lemma} \label{lem:stc} $\stc = \frac{3 \sum_{t \in T} \wgt(t)}{\sum_{e \in E} \wgt(e)}$.
	\end{lemma}
	
	While the following bound is not necessary for our main result, it is instructive 
	to see the largest possible value of the spectral transitivity.
	
	\begin{lemma} \label{lem:ratio-basic}  Consider normalized adjacency matrices $\cA$
		with $n$ vertices. The maximum value of $\stc(\cA)$
		is $1 - 1/(n-1)$. This value is attained for the unique strongly $1$-uniform matrix, 
		the normalized adjacency matrix of the $n$-clique.
	\end{lemma}
	
	\begin{proof} First, consider the normalized adjacency matrix $\cA$ of the $n$-clique.
		All off-diagonal entries are precisely $1/(n-1)$ and $\cA$ can be expressed
		as $(n-1)^{-1}(\bone\bone^T - I)$. The matrix $\cA$ is $1$-regular. The largest eigenvalue is $1$ and all the remaining eigenvalues
		are $-1/(n-1)$. Hence, $\sum_i \lambda^3_i = 1 - (n-1)/(n-1)^3 = 1 - 1/(n-1)^2$.
		The sum of squares of eigenvalue is $\sum_i \lambda^2_i = 1 + (n-1)/(n-1)^2 = 1 + 1/(n-1)$.
		Dividing, 
		$$\frac{\sum_{i \leq n} \lambda^3_i}{\sum_{i \leq n} \lambda^2_i} = 1 - 1/(n-1).$$
		Since the matrix has zero diagonal, the trace $\sum_i \lambda_i$ is zero.
		We will now prove the following claim.
		
		\begin{claim} Consider any sequence of numbers $1 = \lambda_1 \geq \lambda_2 \ldots \geq \lambda_n$
			such that $\forall i, |\lambda_i| \leq 1$ and $\sum_i \lambda_i = 0$. If
			$\sum_i \lambda^3_i \geq (1-1/(n-1)) \sum_i \lambda^2_i$, then $\forall i > 1, \lambda_i = -1/(n-1)$.
		\end{claim}
		
		\begin{proof} Let us begin with some basic manipulations.
			\begin{align}
				\sum_i \lambda^3_i & \geq [1 - 1/(n-1)] \sum_i \lambda^2_i \\ \  \Longrightarrow 1 + \sum_{i > 1} \lambda^3_i &\geq [1 - 1/(n-1)] \cdot (1 + \sum_{i > 1} \lambda^2_i) \nonumber \\
				\Longrightarrow \sum_{i > 1} \lambda^3_i &\geq [1-1/(n-1)] \sum_{i > 1} \lambda^2_i - 1/(n-1). \label{eq:sumeigen}
			\end{align}
			For $i > 1$, define $\delta_i \eqdef \lambda_i + 1/(n-1)$. Note that $\sum_{i > 1} \lambda_i = -1$,
			so $\sum_{i > 1} \delta_i = 0$. Moreover, $\forall i > 1$, $\delta_i \leq 1 + 1/(n-1)$. 
			We plug in $\lambda_i = \delta_i - 1/(n-1)$ in \Eqn{sumeigen}.
			\begin{align*}
				& \sum_{i > 1} \Big[\delta_i - 1/(n-1)\Big]^3 \geq [1-1/(n-1)] \sum_{i > 1} \Big[\delta_i - 1/(n-1)\Big]^2 - 1/(n-1) \\
				\Longrightarrow & \sum_{i > 1} \Big[\delta^3_i - 3\delta^2_i/(n-1) + 3\delta_i/(n-1)^2 - 1/(n-1)^3\Big] \\& \geq [1-1/(n-1)] \sum_{i > 1} \Big[\delta^2_i - 2\delta_i/(n-1)+ 1/(n-1)^2\Big] - 1/(n-1).
			\end{align*}
			Recall that $\sum_{i > 1} \delta_i = 0$. Hence, we can simplify the above inequality.
			\begin{align*}
				& \sum_{i > 1} \delta^3_i - (3/(n-1)) \sum_{i > 1} \delta^2_i - 1/(n-1)^2 \\& \geq [1-1/(n-1)] \sum_{i > 1} \delta^2_i + 1/(n-1) - 1/(n-1)^2 - 1/(n-1) \\
				& \Longrightarrow \sum_{i > 1} \delta^3_i \geq [1+2/(n-1)] \sum_{i > 1} \delta^2_i. \ \ \ \ \textrm{(Canceling terms and rearranging)}
			\end{align*}
			Since $\delta_i \leq (1+1/(n-1))$, we get that $\sum_{i > 1}\delta^3_i \leq [1+1/(n-1)] \sum_{i > 1} \delta^2_i$.
			Combining with the above inequality, we deduce that $[1+2/(n-1)] \sum_{i > 1} \delta^2_i \leq [1+1/(n-1)] \sum_{i > 1} \delta^2_i$.
			This can only happen if $\sum_{i > 1} \delta^2_i$ is zero, implying all $\delta_i$ values are zero. Hence,
			for all $i > 1$, $\lambda_i = -1/(n-1)$.
		\end{proof}
		With this claim, we conclude that any matrix $\cA$ maximizing the ratio of cubes and squares of eigenvalues
		has a fixed spectrum. It remains to prove that a unique normalized adjacency matrix has this spectrum.
		We use the rotational invariance of the Frobenius norm: sum of squares of entries of $\cA$
		is the same as the sum of squares of eigenvalues. Thus,
		\begin{equation}
			\sum_{(u,v) \in E} \frac{2}{d_ud_v} = 1 + \dfrac{ 1}{n-1} = \dfrac{n}{n-1}. \label{eq:clique-frob}
		\end{equation}
		Observe that $\frac{2}{d_ud_v} \geq 1/(d_u(n-1)) + 1/(d_v(n-1))$, since all degrees
		are at most $n-1$. Summing this inequality over all edges,
		\begin{equation}
			\sum_{(u,v) \in E} \frac{2}{d_u d_v} \geq \sum_{v \in V} \sum_{u \in N(v)} \frac{1}{d_v(n-1)} = \sum_{v \in V} \frac{d_v}{d_v(n-1)} = \dfrac{n}{n-1}.
		\end{equation}
		Hence, for \Eqn{clique-frob} to hold, for all edges $(u,v)$, we must have the equality
		$\frac{2}{d_ud_v} = 1/(d_u(n-1)) + 1/(d_v(n-1))$. That implies that for all edge $(u,v)$,
		$d_u = d_v = n-1$. So all vertices have degree $(n-1)$, and the graph is an $n$-clique.
	\end{proof}

	We will need the following ``reverse Markov inequality" for some intermediate proofs. It is easily derivable from first principles. The standard Markov inequality gives us an upper bound, whereas here we derive a similar lower bound for non-negative random variables.
	
	\begin{lemma} \label{lem:rev-markov} Consider a random variable $Z$ taking values in $[0,b]$.
		If $\EX[Z] \geq \sigma b$, then $\Pr[Z \geq \sigma b/2] \geq \sigma/2$.
	\end{lemma}
	
	\begin{proof} In the following calculations, we will upper bound the conditional expectation
		by the maximum value (under that condition).
		\begin{align}
			\sigma b \leq \EX[Z] &= \Pr[Z \geq \sigma b/2] \cdot \EX[Z | Z \geq \sigma b/2] + \Pr[Z \leq \sigma b/2] \cdot \EX[Z | Z \leq \sigma b/2]
			\\ &\leq \Pr[Z \geq \sigma b/2] \cdot b + \sigma b/2
		\end{align}
		We rearrange to complete the proof.
	\end{proof}

\section{The Decomposition Procedure} \label{sec:decomp}

The proof of our main result \Thm{main-decomp} is constructive. We begin with a description of a procedure that outputs the desired decomposition.
The procedure performs an interlacing of ``cleaning" and ``extraction" operations.

There is a single parameter $\eps$, which will be set to $\stc/6$ for the proof of \Thm{main-decomp}. As discussed later,
in implementation we set $\eps$ to some default value.

\begin{algorithm}[ht]
	\caption{\decompose$(G,\eps)$}
	\label{alg:decompose}
	\begin{algorithmic}[1]
		\State Initialize $\bX$ to be an empty family of sets, and initialize subgraph $H = G$.
		\While{$H$ is non-empty} 
		\While{$H$ is not clean}
		\State Remove any edge $e \in E_H$ from $H$ such that $\wgt(T_H(e)) < \stch \cdot \wgt(e)$. \label{step:clean}
		\EndWhile
		\State Add output \extract$(H, \eps)$ to $\bX$.
		\State Remove these vertices from $H$.
		\EndWhile
		\State Output $\bX$.
	\end{algorithmic}
\end{algorithm}

\begin{algorithm}[ht]
	\caption{\extract$(H, \eps)$}
	\label{alg:extract}
	\begin{algorithmic}[1]
		\State Pick $v\in V_H$ that minimizes $d_v$.
		\State Construct the set $L \eqdef \{u | (u,v) \in E_H, d_u \leq 2\stch^{-1}d_v\}$ ($L$ is the set of low degree neighbors of $v$ in $H$.)
		\State For every vertex $w \in V_H$, define $\rho_w$ to be the total weight of triangles of the form $(w,u,u')$ where $u, u' \in L$.
		\State Sort the vertices in decreasing order of $\rho_w$, and construct the ``sweep cut" $C$ to be
		the smallest set satisfying $\sum_{w \in C} \rho_w \geq (1/2) \sum_{w \in V_H} \rho_w$.
		\State Output $X \eqdef \{v\} \cup L \cup C$.
	\end{algorithmic}
\end{algorithm}

A central notion to understand the above procedure is the notion of a clean graph.

\begin{definition} \label{def:clean} 
A connected subgraph $H$ is called \emph{clean} if $\forall e \in E_H$, $\wgt(T_H(e)) \geq \stch{} \wgt(e)$.
\end{definition} 

In the \decompose{} procedure, we repeatedly remove edges until the remaining graph is clean. Note that
we do not specify any ordering on this removal; all we need is for the resulting $H$ to be clean.
We then call \extract$(H, \eps)$ to remove a single set $X$, which forms one of the sets in \Thm{main-decomp}.
This process is iterated until the graph is empty.

There are two main challenges in the proof. First, we need to argue that $\cA|_X$ is strongly uniform,
for every $X$ extracted. This corresponds to cluster structure in \Thm{main-decomp}. Second, we have to argue that the ``damage" 
done by the extractions and cleaning is limited. This corresponds to coverage in \Thm{main-decomp}. 
When $X$ is extracted, the incident edges (and triangles) not contained in $X$ are effectively deleted.
When an edge is cleaned in \Step{clean}, it is removed from the graph. The corresponding Frobenius norm is lost,
and hence we need to upper bound this loss.

The most difficult part of the proof is to understand the effects of extraction.
We can prove that, roughly speaking, the triangle weight extracted in $X$
is proportional to the total triangle weight incident to $X$. This proof crucially uses the fact
that the current graph $H$ (where \extract$(H,\eps)$ is called) is clean. 

By itself, the above bound could be trivially obtained, simply by having $X$ be the whole vertex set.
The challenge is to also prove that $\cA|_X$ is (strongly) uniform. Observe that $X$ is a set
of radius $2$, since the extraction start at vertex $v$, take a low-degree neighborhood $L$,
and then takes some vertices $w$ that form many triangles involving $L$. Hence, any such $w$
is a neighbor of $L$ and at most distance $2$ from the starting vertex $v$. We need to argue
that this distance $2$ set of vertices is small, so the triangle weight extracted/contained in $X$
is inside a small set of vertices.

The above arguments about the \extract{} procedure are contained in \Sec{extract}. In \Sec{wrapup}, we also
bound the losses from cleaning and wrap up the entire proof. Those calculations are much simpler.

\section{Cleaned graphs and extraction} \label{sec:extract}

Consider a subgraph $H$ that is connected clean, and let $X$ denote the output of \extract$(H,\eps)$.
(Recall that $\stch$ is set to $\stc/6$.) The main theorem of this section follows.
We do not attempt to optimize the polynomials in $\stch$, and note that the dependence may be much
better than the bounds stated.

\begin{theorem} \label{thm:extract} 
    $$ \sum_{t \in T_H, t \subseteq X} \wgt(t) \geq \poly(\stch) \sum_{t \in T_H, t \cap X \neq \emptyset} \wgt(t)$$
	(The triangle weight contained inside $X$ is a constant fraction of the triangle weight incident to $X$.)
	
    Moreover, $\cA|_X$ is strongly $\poly(\stch)$-uniform.
\end{theorem}

We use $v$, $L$, and $C$ as defined in \extract$(H,\eps)$. We use $N(v)$ (or $N$ when the context is clear) to denote the set of neighbors of a vertex $v$.
The proof of this theorem contains numerous parts. We first state useful lemmas that bound
edge/triangle weights incident to $L$ and $C$. 
These lemmas contain the core calculations that bound the total triangle weight incident to
the extracted set, and upper bound the size of the extracted set. From the various bounds
of these two subsections, the final proof of \Thm{extract} follows with some calculations.

Our first step is prove that the total edge weight contained in $L$ is sufficiently large.

\begin{lemma} \label{lem:nbd-wgt} $\sum_{e \in E_H, e \subseteq L} \wgt(e) \geq \stch^2/8$.
\end{lemma}

\begin{proof} For any vertex $u \in N$, we define the set of partners $P(u)$ to be
$\{w: (u,v,w) \in T_H\}$. 
We first prove the bound 
    \begin{equation} \label{eq:recip}
    \sum_{w \in P(u) \cap L} d^{-1}_w \geq \stch/2
\end{equation}
Let $e = (u,v)$. Since $H$ is clean, $\wgt(T_H(e)) \geq \stch \wgt(e)$.
Expanding out the definition of weights,
\begin{equation} \label{eq:partner}
	\sum_{w: (u,v,w) \in T_H} \frac{1}{d_u d_v d_w} \geq \frac{\stch}{d_ud_v} \ \ \ \Longrightarrow \ \ \ \sum_{w \in P(u)} {d^{-1}_w} \geq \stch. 
\end{equation}
Note that $L$ (as constructed in \extract$(H)$) is the subset
of $N$ consisting of vertices with degree at most $2\stch^{-1}d_v$. 
For $w \in N \setminus L$, we have the lower bound $d_w \geq 2\stch^{-1}d_v$. Hence,
\begin{equation} \label{eq:highdeg}
	\sum_{w \in N \setminus L} d^{-1}_w \leq |N \setminus L| (\stch/2) d^{-1}_v \leq d_v \times (\stch/2) d^{-1}_v = \stch/2.
\end{equation}

In the calculation below, we split the sum of \Eqn{partner} into the contribution from $L$ and from outside $L$.
We apply \Eqn{highdeg} to bound the latter contribution.
\begin{equation}
	\stch \leq \sum_{w \in P(u)} d^{-1}_w \leq \sum_{w \in P(u) \cap L} d^{-1}_w + \sum_{w \in N \setminus L} d^{-1}_w \leq \sum_{w \in P(u) \cap L} d^{-1}_w + \stch/2.
\end{equation}

By rearranging, we prove \Eqn{recip}.

We now use \Eqn{recip} to complete the proof.
By \Eqn{recip}, $\forall w \in L$, $\sum_{w' \in P(w) \cap L} d^{-1}_{w'} \geq \stch/2$. 
We multiply both sides by $d^{-1}_{w}$ and sum over all $w \in L$. 
\begin{equation}
	\sum_{w \in L} \sum_{w' \in P(w) \cap L} (d_w d_{w'})^{-1} \geq (\stch/2) \sum_{w \in L} d^{-1}_{w}.
\end{equation}
By \Eqn{recip}, $\sum_{w \in L} d^{-1}_{w} \geq \stch/2$. Note that $w' \in P(w)$ only if $(w,w') \in E_H$.
Hence, 

\noindent
$\sum_{w \in L} \sum_{w' \in L, (w,w') \in E_H} \wgt((w,w')) \geq \stch^2/4$. Note that the summation
counts all edges twice, so we divide by $2$ to complete the proof.
\end{proof}

Our next step is to bound the total triangle weight of triangles that involve
at least two vertices of $L$.  Recall, from the description of \extract, 
that $\rho_w$ is the total triangle weight of the triangles $(w,u,u')$, where $u,u' \in L$.
We will prove that $\sum_w \rho_w$ is large, which is quite 
easy using the bound of \Lem{nbd-wgt}.
	
\begin{claim} \label{clm:rho-bound} $\sum_{w \in V_H} \rho_w \geq \stch^3/8$.
\end{claim}

\begin{proof} Note that $\sum_{w \in V_H} \rho_w$ is equal to $\sum_{e \in E_H, e \subset L} \wgt(T_H(e))$.
	Both these expressions give the total weight of all triangles in $H$ that involve
	two vertices in $L$. Since $H$ is clean, for all edges $e \in E_H$, $\wgt(T_H(e)) \geq \stch \wgt(e)$.
	Hence, $\sum_{e \in E_H, e \subset L} \wgt(T_H(e)) \geq \stch \sum_{e \in E_H, e \subset L} \wgt(e)$.
    Applying \Lem{nbd-wgt}, we can lower bound the latter by $\stch^3/8$.
\end{proof}

At this stage, we have lower bounds on the edge weight contained in $L$ and the triangle
weight incident to (edges in) $L$. We come to a central calculation; upper bounding the size
of $C$. The sweep cut step in \extract{} removes half of the triangle weight incident to $L$.
We actually show that the cut size is quite small and comparable to $L$ (or the degree of 
the starting vertex $v$). This is quite surprising, since the triangle weight 
incident to $L$ can be ``spread" over a large set of vertices in the graph. Nonetheless,
we show that at least half of this weight is concentrated on a small set.
(We do not try to optimize the constant factors.)

\begin{lemma} \label{lem:C} $|C| \leq 144\stch^{-5} d_v$.
\end{lemma}

\begin{proof} We first show that a few $\rho_w$ values dominate the sum, using a somewhat roundabout argument.
This bound will help us show that $C$ is small.

We will prove the following bound:

\begin{equation} \label{eq:rho-sqrt}
    \sum_{w \in V_H} \sqrt{\rho_w} \leq 2 \stch^{-1} \sqrt{d_v}.
\end{equation}

Let $c_w$ be the number of vertices in $L$ that are neighbors (in $H$) of $w$. Note that
for any triangle $(u,u',w)$ where $u,u' \in L$, both $u$ and $u'$ are common neighbors of $w$ and $v$. 
The number of triangles $(u,u',w)$ where $u,u' \in L$ is at most $c^2_w$.
The weight of any triangle in $H$ is at most $d^{-3}_v$, since $d_v$ is the lowest degree (in $G$)
of all vertices in $H$. As a result, we can upper bound $\rho_w \leq d^{-3}_v c^2_w$.

Taking square roots and summing over all vertices,
\begin{equation} \label{eq:sqrt}
	\sum_{w \in V_H} \sqrt{\rho_w} \leq d^{-3/2}_v \sum_{w \in V_H} c_w .
\end{equation}
Note that $\sum_{w \in V_H} c_w$ is exactly the sum over $u \in L$ of the degrees
of $u$ in the subgraph $H$. (Every edge incident to $u \in L$ gives a unit contribution
to the sum $\sum_{w \in V_H} c_w$.) By definition, every vertex in $L$ has
degree in $H$ at most $2\stch^{-1}d_v$. The size of $L$ is at most $d_v$.
Hence, $\sum_{w \in V_H} c_w \leq 2\stch^{-1} d^2_v$. Plugging into \Eqn{sqrt},
we deduce that $\sum_{w \in V_H} \sqrt{\rho_w} \leq 2\stch^{-1} \sqrt{d_v}$.

\medskip

Now that \Eqn{rho-sqrt} is proven, we work on bounding $|C|$.
For convenience, let us reindex vertices so that $\rho_1 \geq \rho_2 \geq \rho_3 \ldots$.
Let $r \leq n$ be an arbitrary index. Because we index in non-increasing order, note that $\sum_{j \leq n} \rho_j \geq r \rho_r$.
Furthermore, $\forall j > r$, $\rho_j \leq \rho_r$.
\begin{eqnarray} \label{eq:rho-sum}
	\sum_{j > r} \rho_j \leq \sqrt{\rho_r} \sum_{j > r} \sqrt{\rho_j} \leq \sqrt{\frac{\sum_{j \leq n} \rho_j}{r}} \sum_{j \leq n} \sqrt{\rho_j}
	= \Big[\frac{\sum_{j \leq n} \sqrt{\rho_j}}{\sqrt{r} \cdot \sqrt{\sum_{j \leq n} \rho_j}}\Big] \sum_{j \leq n} \rho_j
\end{eqnarray}
Observe that \Eqn{rho-sqrt} gives an upper bound on the numerator, while \Clm{rho-bound} gives a lower bound
on (a term in) the denominator. Plugging those bounds in \Eqn{rho-sum},
\begin{equation}
	\sum_{j > r} \rho_j \leq \frac{2\stch^{-1}\sqrt{d_v}}{\sqrt{r} \cdot \stch^{3/2}/\sqrt{8}} \sum_{j \leq n} \rho_j \leq \frac{1}{\sqrt{r}} \cdot \frac{6\sqrt{d_v}}{\stch^{5/2}} \cdot \sum_{j \leq n} \rho_j.
\end{equation}
Suppose $r > 144 \stch^{-5} d_v$. Then $\sum_{j > r} \rho_j < (1/2) \sum_{j \leq n} \rho_j$. The sweep cut $C$ is 
constructed with the smallest value of $r$ such that $\sum_{j > r} \rho_j < (1/2) \sum_{j \leq n} \rho_j$. Hence,
$|C| \leq 144 \stch^{-5} d_v$.
\end{proof}

We state an additional technical claim that bounds the triangle weight incident to a single vertex.

\begin{claim} \label{clm:tri-wgt-vertex} For all vertices $u \in V_H$, $\wgt(T_H(u)) \leq (2d_v)^{-1}$.
\end{claim}

\begin{proof} Consider edge $(u,w) \in E_H$.
	We will prove that $\wgt(T_H((u,w))) \leq d^{-1}_u d^{-1}_v$. Recall that $d_v$ is the smallest degree
	among vertices in $H$. Furthermore, $|T_H((u,w))| \leq d_w$, since the third vertex in a triangle containing $(u,w)$
	is a neighbor of $w$.
	$$ \wgt(T_H((u,v))) = \sum_{z: (z,u,w) \in T_H} \frac{1}{d_u d_w d_z} \leq \frac{1}{d_u d_v} \sum_{z: (z,u,w) \in T_H} \frac{1}{d_w}
	\leq \frac{1}{d_u d_v} \times \frac{d_w}{d_w} = \frac{1}{d_u d_v}$$
	We now bound $\wgt(T_H(u))$ by summing over all neighbors of $u$ in $H$.
	\begin{align*}
		\wgt(T_H(u)) = (1/2) \sum_{w: (u,w) \in E_H} \wgt(T_H((u,w))) \leq (1/2) \sum_{w: (u,w) \in E_H} \frac{1}{d_u d_v}
		= \frac{1}{2d_v} \sum_{w: (u,w) \in E_H} \frac{1}{d_u} \leq \frac{1}{2d_v} \times \frac{d_u}{d_u} = \frac{1}{2d_v}. 
	\end{align*}
\end{proof}

\subsection{The proof of \Thm{extract}} \label{sec:proof-extract}
	
Recall that $X$ is $\{v\} \cup L \cup C$.
By construction, the total weight of triangles inside $X$ is at least $\sum_{v \leq n} \rho_v/2$.
By \Clm{rho-bound}, $\sum_{v \leq n} \rho_v/2 \geq \stch^3/16$. 

Let us now bound that total triangle weight incident to $X$ in $H$. Observe that $|X| = 1 + |L| + |C|$
which is at most $1 + d_v + \stch^{-5}144 d_v$, by \Lem{C}. We can further bound $|X| \leq \stch^{-5}146 d_v$.
By \Clm{tri-wgt-vertex}, the total triangle weight incident to a vertex is at most $(2d_v)^{-1}$.
Multiplying, the total triangle weight incident to all of $X$ is at most $73 \stch^{-5}$.

Thus, the triangle weight contained in $X$ is at least $\frac{\stch^3/16}{73 \stch^{-5}} = \Omega(\stch^8)$ times
the triangle weight incident to $X$. This completes the proof of the first statement of \Thm{extract}.

{\bf Proof of uniformity of $\cA|_X$:} 
For convenience, let $B$ denote the set $\{e | e \subseteq L, \wgt(e) \geq \stch^2d^{-2}_v/16\}$.
By \Lem{nbd-wgt}, $\sum_{e \subseteq L} \wgt(e) \geq \stch^2/8$.
There are at most ${d_v \choose 2} \leq d^2_v/2$ edges in $B$. For every edge $e$, $\wgt(e) \leq 1/d^2_v$. 
\begin{align*}
	\frac{\stch^2}{8} \leq \sum_{e \notin B} \wgt(e) + \sum_{e \in B} \wgt(e) 
	\leq d^2_v \times \stch^2d^{-2}_v/16 + |B|d^{-2}_v = \stch^2/16 + |B|d^{-2}_v.
\end{align*}
Rearranging, $|B| \geq \stch^2d^2_v/16$. 

Hence, there are at least $\stch^2d^2_v/16$ edges contained in $X$
with weight at least $\stch^2d^{-2}_v/16$. Consider the random variable $Z$ that
is the weight of a uniform random edge contained in $X$. Since $|X| \leq \stch^{-5} 144 d_v$,
the number of edges in $X$ is at most $\stch^{-10} (144)^2 d^2_v$. So, 
\begin{equation}
    \EX[Z] \geq \frac{\stch^2 d^2_v/16}{\stch^{-10}(144)^2 d^2_v} \times \stch^2d^{-2}_v/16 = \Omega(\stch^{14} d^{-2}_v).
\end{equation}
The maximum value of $Z$ is the largest possible weight of an edge in $E_H$,
which is at most $d^{-2}_v$. Applying the reverse Markov bound of \Lem{rev-markov},
with probability $\Omega(\stch^{14})$, a uniform random edge in $X$
has weight $\Omega(\stch^{14} d^{-2}_v)$.

By \Lem{nbd-wgt} $\sum_{e \in E_H, e \subseteq L} \wgt(e) \geq \stch^2/8$.
The weight of any edge is at most $d^{-2}_v$, so there must be at least $\stch^2 d^2_v/8$ edges
in $L$. This implies that $|L| = \Omega(\stch d_v)$. Thus, $|X| = \Omega(\stch d_v)$
and $d^{-2}_v = \Omega(\stch^2/|X|^2)$. 

With the bounds from the previous two paragraphs, there is an $\Omega(\stch^{14})$
fraction of edges in $X$ whose weight is $\Omega(\poly(\stch)/|X|^2)$. 
So we prove the uniformity of $\cA|_X$. 

{\bf Proof of strong uniformity:} For strong uniformity, we need to repeat the above argument
within neighborhoods in $X$. We prove in the beginning of this proof that the total triangle weight inside $X$ is at least $\stch^3/16$. 
We also proved that $|X| \leq 146\stch^{-5}d_v$. Consider the random variable $Z$ that is the triangle
weight contained in $X$ incident to a uniform random vertex in $X$. Note that $\EX[Z] \geq (\stch^3/16) / (146\stch^{-5} d_v) = \Omega(\stch^8 d^{-1}_v)$.
By \Clm{tri-wgt-vertex}, $Z$ is at most $(2d_v)^{-1}$. Applying \Lem{rev-markov},
at least $\Omega(\stch^8 |X|)$ vertices in $X$ are incident to at least $\Omega(\stch^8 d^{-1}_v)$ triangle weight
inside $X$.

Consider any such vertex $u$. Let $N(u)$ be the neighborhood of $u$ in $X$. 
Each triangle in $X$ has weight at most $d^{-3}_v$. The triangle weight
incident to $u$ is at most $d^{-3}_v|N(u)|^2$. By the choice of $u$, this triangle
weight is at least $\Omega(\stch^8 d^{-1}_v)$. Hence, $|N(u)| = \Omega(\stch^4 d_v)$.

Every edge $e$
in $N(u)$ forms a triangle with $u$ with weight $\wgt(e)/d_u$. Hence, noting that $d_u \geq d_v$,
\begin{equation}
    \sum_{e \subseteq N(u)} \wgt(e) d^{-1}_u = \Omega(\stch^8 d^{-1}_v) \ \ \ \Longrightarrow \ \ \ \sum_{e \subseteq N(u)} \wgt(e) = \Omega(\stch^8).
\end{equation}
There are at most $|X|^2 \leq \stch^{-10} (146)^2 d^2_v$ edges in $N(u)$. Let $Y$ denote the weight
of a uniform random edge in $N(u)$. Note that $\EX[Y] = \Omega(\stch^8/(\stch^{-10} (146)^2 d^2_v)) = \Omega(\stch^{18} d^{-2}_v)$.
The maximum weight of an edge is at most $d^{-2}_v$. By \Lem{rev-markov}, 
at least an $\Omega(\stch^{18})$ fraction of edges in $N(u)$ have a weight of at least $\Omega(\stch^{18} d^{-2}_v)$.
Since $|N(u)| = \Omega(\stch^4 d_v)$, $N(u)$ is also $\poly(\stch)$-uniform.

\section{Wrapping up} \label{sec:wrapup}

We complete the proof of our main result, \Thm{main-decomp}. This proof mainly involves
putting together the various claims and lemmas from the previous section.

\begin{proof} (of \Thm{main-decomp}) 
We partition all the triangles of $G$ into three sets depending on how they
are affected by \decompose$(G)$.
(i) The set of triangles removed by the cleaning
step of \Step{clean}, (ii) the set of triangles contained in some $X_i \in \bX$, or (iii) 
the remaining triangles. Abusing notation, we refer to these sets as $T_C$, $T_X$, and $T_R$ respectively. For each such $T_i$, we all call $\wt(T_i)$ the sum of the weights of all triangles in the set.
Note that the triangles of $T_R$ are the triangles ``cut" when $X_i$ is removed.

Let us denote by $H_1, H_2, \ldots, H_k$ the subgraphs of which \extract{} is called.
Let the output of \extract$(H_i)$ be denoted $X_i$.
By the uniformity guarantee of \Thm{extract}, each $\cA|_{X_i}$ is $\delta \stc^c$-uniform for appropriate $\delta$ and a constant $c$.

It remains to prove the coverage guarantee. We now sum the bound of \Thm{extract}
over all $X_i$. (For convenience, we expand out $\stch$ as $\stc/6$ and let $\delta'$ denote
a sufficiently small constant.)
\begin{equation}
	\sum_{i \leq k} \sum_{t \in T(H_i), t \subseteq X_i} \wgt(t) \geq (\delta' \stc^8) \sum_{i \leq k} \sum_{t \in T(H_i), t \cap X_i \neq \emptyset} \wgt(t).
\end{equation}
The LHS is precisely $\wgt(T_X)$. Note that a triangle appears at most once in the double summation 
in the RHS. That is because if $t \cap X_i \neq \emptyset$, then $t$ is removed
when $X_i$ is removed. Since $H_i$ is always clean, the triangles of $T_C$ cannot participate
in this double summation. Hence, the RHS summation is $\wgt(T_X) + \wgt(T_R)$
and we deduce that 
\begin{equation} \label{eq:wgt-bound}
	\wgt(T_X) \geq \delta' \stc^8 (\wgt(T_X) + \wgt(T_R))
\end{equation}

Note that $\wgt(T_c) + \wgt(T_x) + \wgt(T_r) = \sum_{t \in T} \wgt(t)$. There is where
the definition of $\stc$ makes its appearance. By \Lem{stc}, we can write the above equality
as $\wgt(T_C) + \wgt(T_X) + \wgt(T_R) = (\stc/3) \sum_{e \in E} \wgt(e)$. 

We now prove the bound $\wgt(T_C) \leq (\stc/6) \sum_{e \in E} \wgt(e)$.
Consider an edge $e$ removed at \Step{clean} of \decompose. Recall that $\stch$ is set to $\stc/6$. At that removal,
the total weight of triangles removed (cleaned) is at most $(\stc/6) \wgt(e)$. An edge
can be removed at most once, so the total weight of triangles removed by cleaning
is at most $(\stc/6) \sum_{e \in E} \wgt(e)$.

Hence, we get the inequality $\wgt(T_X) + \wgt(T_R) \geq (\stc/6) \sum_{e \in E} \wgt(e)$.
By \Clm{frob}, $\sum_{e \in E} \wgt(e) = \|\cA\|^2_F/2$. Applying \Eqn{wgt-bound},
$\wgt(T_X) = \poly(\stc)\|\cA\|^2_F$. By \Clm{tri-frob},
$\sum_{i \leq k} \|\cA|_{X_i}\|^2_F = \Omega(\wgt(T_X))$, which is at least $\poly(\stc) \|\cA\|^2_F$.
That completes the proof of the coverage bound.

\end{proof}

\subsection{Algorithmics} \label{sec:impl}

We discuss implementations of the procedures computing the decomposition
of \Thm{main-decomp}. The main operation required is a triangle enumeration of $G$; there
is a rich history of algorithms for this problem. The best known bound for sparse graph
is the classic algorithm of Chiba-Nishizeki that enumerates all triangles in $O(m\alpha)$ time,
where $\alpha$ is the graph degeneracy.

We provide a formal theorem providing a running time bound. We do not explicitly
describe the implementation through pseudocode, and instead explain the main details
in the proof. In practice, we implement the algorithm described in the proof, and 
observe it to have good empirical performance.

\begin{theorem} \label{thm:runtime} There is an implementation of \decompose$(G)$ whose 
running time is $O(R + (m + n + T)\log n)$, where $R$ is the running time of listing
all triangles. The space required is $O(T)$ (where $T$ is the triangle count).
\end{theorem}

\begin{proof} We assume an adjacency list representation where each list is stored
in a dictionary data structure with logarithmic time operations (like a self-balancing 
binary tree).

We prepare the following data structure that maintains
information about the current subgraph $H$. We initially set $H = G$.
We will maintain all lists as hash tables so that elementary operations
on them (insert, delete, find) can be done in $O(1)$ time.

\begin{compactenum}
	\item A list of all triangles in $T_H$ indexed by edges. Given an edge $e$,
	we can access a list of triangles in $T_H$ containing $e$.
	\item A list of $\wgt(T_H(e))$ values for all edges $e \in E_H$.
	\item A list $U$ of all (unclean) edges such that $\wgt(T_H(e)) < \stch \wgt(e)$.
	\item A min priority queue $Q$ storing all vertices in $V_H$ keyed by degree $d_v$. We will
	assume pointers from $v$ to the corresponding node in $Q$.
\end{compactenum}

\medskip

These data structures can be initialized by enumerating all triangles, indexing them, and preparing all the lists.
This can be done in $O(R)$ time. 

We describe the process to remove an edge from $H$. When edge $e$ is removed,
we go over all the triangles in $T_H$ containing $e$. For each such triangle $t$
and edge $e' \in t$, we remove $t$ from the triangle list of $e'$. We then update
$\wgt(T_H(e'))$ by reducing it by $\wgt(t)$. If $\wgt(T_H(e'))$ is less than $\wgt(e)$,
we add it to $U$. Finally, if the removal of $e$ removes a vertex $v$ from $V_H$ (i.e. $v$ has no edges incident to it remaining in $H$),
we remove $v$ from the priority queue $Q$. Thus, we can maintain the data structures.
The running time is $O(|T_H(e)|)$ plus an additional $\log n$ for potentially updating $Q$.
The total running time for all edge deletes is $O(T + n\log n)$.

With this setup in place, we discuss how to implement \decompose. The cleaning operation in \decompose{}
can be implemented by repeatedly deleting edges from the list $U$, until it is empty.

We now discuss how to implement \extract. We will maintain a max priority queue $R$ maintaining
the values $\{\rho_w\}$. Using $Q$ as defined earlier, we can find the vertex $v$ of minimum degree.
By traversing its adjacency list in $H$, we can find the set $L$. We determine all edges in $L$
by traversing the adjacency lists of all vertices in $L$. For each such edge $e$, we enumerate
all triangles in $H$ containing $e$. For each such triangle $t$ and $w \in t$, we will update 
the value of $\rho_w$ in $R$. 

We now have the total $\sum_w \rho_w$ as well. We find the sweep cut by repeatedly deleting from the
max priority queue $R$, until the sum of $\rho_w$ values is at least half the total. Thus, we can compute 
the set $X$ to be extracted. The running time is $O((|X| + |E(X)| + |T(X)|)\log n)$, where $E(X), T(X)$
are the set of edges and triangles incident to $X$.

Overall, the total time for all the extractions and resulting edge removals
is $O((n+m+T)\log n)$. The initial triangle enumeration takes $R$ time. We add to complete the proof.
\end{proof}
\section{Empirical Validation} \label{sec:emp}
We present an empirical validation of \Thm{main-decomp} and the procedure \decompose.
We show that spectral triadic decompositions exist in real-world networks and are semantically meaningful.
We perform experiments on a number of real-world networks, listed in \Tab{summary}.

\paragraph{Datasets:} The datasets have been taken from SNAP~\cite{snapnets}, the Network Repository~\cite{nr}, Arnetminer~\cite{aminer}, and Newman’s graph collection~\cite{NewmanCondMat99}. 
The two coauthorship graphs come with metadata. The ca-DBLP graph has manually curated ground truth clusters based on discipline of authors, and the ca-cond-matL graph is labeled by the name of 
the author/scientist. 
We use these datasets to determine the semantic meaning of the extracted clusters.
We make the graphs undirected. The network names indicate the raw data source: `ca' refers to coauthorship networks (ca-CondMat is for researchers who work in condensed matter, ca-DBLP does the same for researchers whose work is on DBLP, a computer science bibliography website), ones beginning with `soc' are social networks (socfb-Rice31 is a Facebook network, soc-hamsterster is from Hamsterster, a pet social network), and `cit' refers to citation networks.

\begin{table}
	\begin{center}
		\begin{tabular}{||c|| c |c| c|c||} 
			\hline
			Dataset & \#Vertices & \#Edges & \#Triangles & $\tau$  \\ [0.5ex] 
			\hline\hline
			soc-hamsterster & 2,427 & 16,630 & 53,251 & 0.215 \\ 
			\hline
			socfb-Rice31 & 4,088 & 184,828 & 1,904,637 & 0.122\\ 
			\hline
			caHepTh & 9,877 & 24,827 & 28,339 &0.084\\
			\hline
			ca-cond-matL & 16,264 & 47,594 & 68,040 & 0.255 \\
			\hline
			ca-CondMat & 23,133  & 93,497 & 176,063 &0.125 \\
			\hline
		    cit-DBLP & 217,312 & 632,542 & 248,004 & 0.087\\
			\hline
			ca-DBLP & 317,080 & 1,049,866 & 2,224,385  & 0.248\\ [1ex]

			\hline
		\end{tabular}\caption{Summary of datasets. } \label{tab:summary}
	\end{center}
\end{table}

\paragraph{Implementation details:} The code is written in Python, and we run it on Jupyter using Python 3.7.6 on a Dell notebook with an Intel i7-10750H processor and 32 GB of ram. The code requires enough storage to store all lists of triangles, edges and vertices, and may be found on github at \url{https://bitbucket.org/Boshu1729/triadic/src/master/}. We set the parameter $\eps$ to $0.1$ for all the experiments, unless stated otherwise. 
In general, we observe that the results are stable with respect to this parameter, and it is a convenient
choice for all datasets.

\paragraph{The $\tau(G)$ values of real data:} In \Tab{summary}, we list the spectral triadic content,
$\tau(G)$, of the real-world networks. Observe that they are quite large. They are the highest in social networks, consistently ranging in values greater than $0.1$. 
This shows the empirical significance of $\tau(G)$ in real-world networks, which is consistent with large
clustering coefficients. 

\subsection{Cluster details} \label{sec:details-exp}

A spectral triadic decomposition produces a large number of approximately
uniform dense clusters, starting from only the promise of a large $\stc(G)$ value. We compute
these decompositions for all our datasets, using the algorithm in the proof of \Thm{runtime}. For comparison, we also compute
the clustering output of the classic Louvain algorithm~\cite{Louvain}.
In \Fig{density},
we give a scatter point of clusters with axes of cluster size versus edge density. (Edge density is
the number of edges divided by the number of vertices choose 2.) Each cluster is a point in the plot.
In \Tab{datasummary}, we give a summary of cluster properties.

\begin{table}[ht]
\begin{center}

	\begin{tabular}{||c||c|c|c|c|c|c|}
		\hline 
		Dataset & \# Clusters & Frac. of$\|M\|_F$ \% & Largest Cl. & \% $V$ in Cl. & Mean ED & 10$^{th}$ \% ED\\ \hline
		soc-hamsterster & 208 & 85.34 & 81 & 76.1\% & 0.79 & 0.33 \\ \hline
		socfb-Rice31 & 86 & 36.76 & 230 & 86.8\% & 0.46 & 0.19 \\ \hline
		ca-HepTh & 849 &  73.79 & 47 & 58.1\% & 0.72 & 0.33 \\ \hline
		ca-CondMat & 2049 & 58.84 & 68 & 75.6\% & 0.70 & 0.30\\ \hline
		ca-condmatL & 1566 & 78.64 &47 & 71.1\% & 0.71 & 0.32\\ \hline
		cit-DBLP & 7265 & 77.15 & 111 & 27.6\% & 0.49 & 0.19\\ \hline
	\end{tabular}
	
	\caption{Summary of spectral triadic clusters across datasets: number of clusters, fraction of Frobenius norm, the size of the largest cluster, total number of vertices in clusters, mean edge density, and 10$^{th}$ percentile of density.}\label{tab:datasummary}
\end{center}
\vspace{-20pt}
\end{table}

We observe that in all cases, the decompositions create a large number of dense clusters. These densities are typically more than $0.2$, which is quite high.
For four of the six datasets we studied, the clusters preserved over 70\% of the vertices. We also show that fraction of total Frobenius norm
contained in these clusters, and observe that it is typically more than 50\%. The clusters are never too large, which is understandable given the 
large edge density.

For comparison, the scatter plots in \Fig{density} also give the clusters of the classic Louvain decomposition~\cite{Louvain}. While Louvain gives
numerous clusters, in almost all cases, it gives extremely large clusters that are of low density (towards the bottom right
of the plots). For cit-DBLP, the Louvain output is extremely sparse, in comparison to the spectral triadic clusters.
In general, the spectral triadic clusters are above in the plots, meaning that they
are denser than the clusters produced by the Louvain algorithm. 
\begin{wrapfigure}{r}{0.35\linewidth}
	\centering
	\includegraphics[width = 0.8\linewidth]{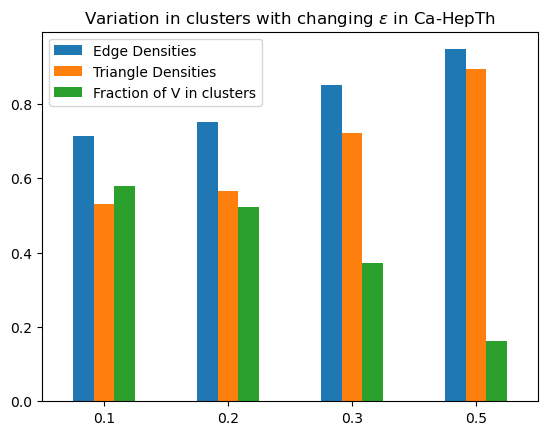}
	\caption{Mean edge density, triangle density, and fraction of vertices preserved in clusters extracted from caHepTh for $\eps=0.1, 0.2, 0.3, 0.5$. Edge and triangle density rise with increasing $\eps$, but fewer vertices are preserved.}\label{fig:eps}
\end{wrapfigure}
For a deeper understanding, we show a scatter plot for uniformity values in \Fig{unif}. (For ease of reading,
we only provide plots on three datasets, but our results are consistent across all of them.) The uniformity
of a cluster is the largest possible $\alpha$ value according to \Def{uniform}. Recall that large uniformity
means that cluster is more ``clique-like", and involves vertices whose degree is comparable to cluster size.
We see that the spectral triadic clusters typically have high uniformity, in constrast to Louvain clusters.
So while Louvain clusters may be somewhat dense (according to \Fig{density}), these clusters involve high degree vertices (so the cluster has low uniformity). Thus, these clusters are less 
community-like than spectral triadic clusters. We will see more evidence of this in our semantic experiments.

\begin{figure*}[h!]
	\centering
	\begin{subfigure}[soc-hamsterster]
		{\includegraphics[width=0.3\textwidth]{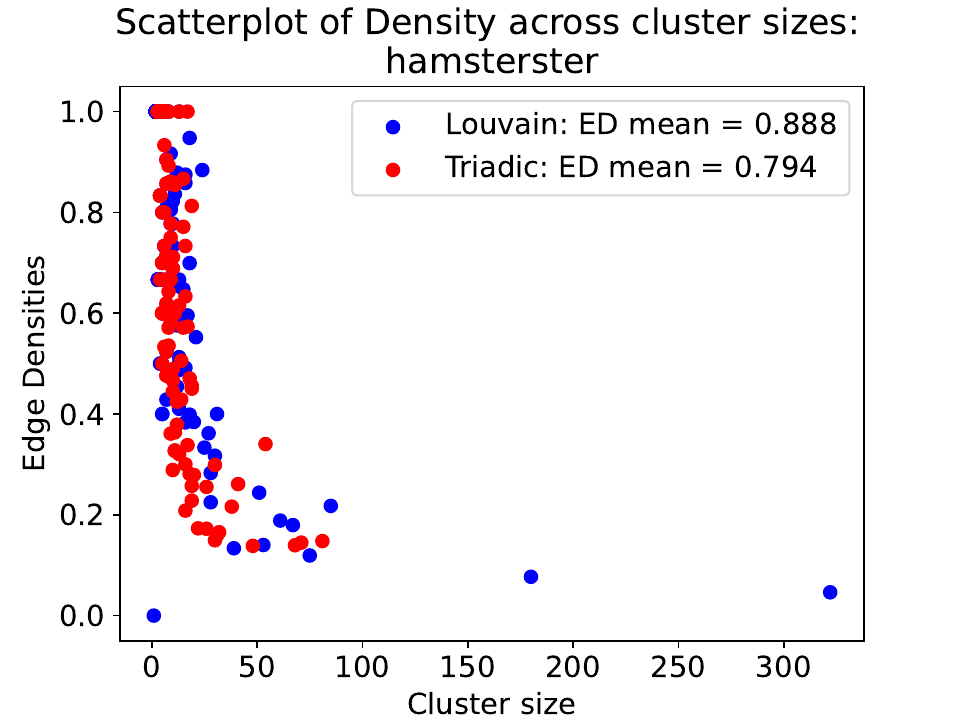}}
	\end{subfigure}	
	\begin{subfigure}[socfb-Rice31]
		{\includegraphics[width=0.3\textwidth]{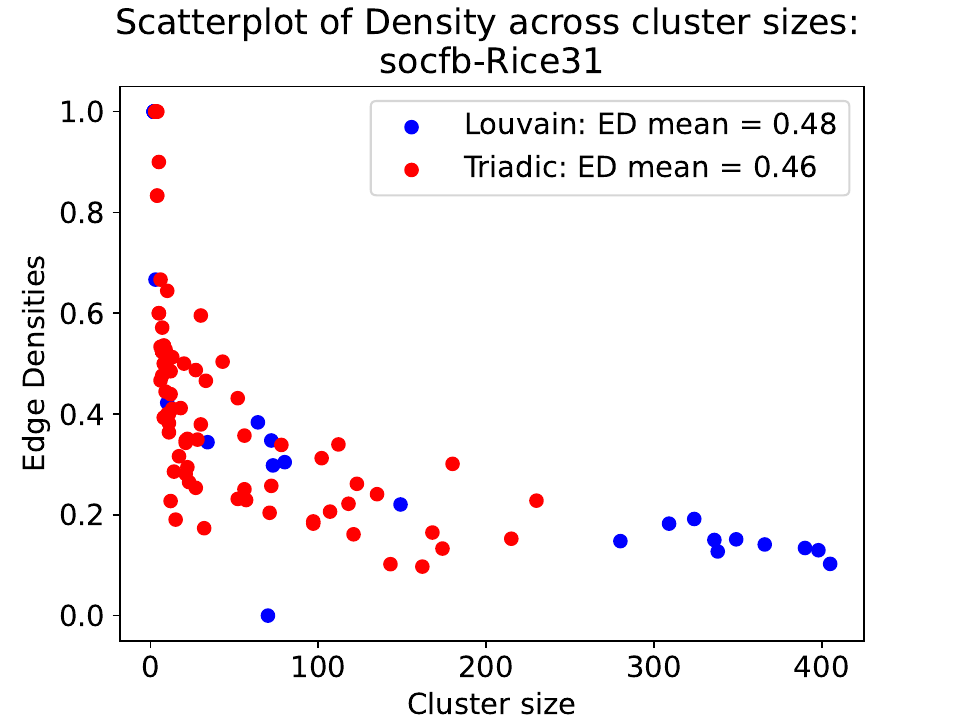}}
	\end{subfigure}
    \begin{subfigure}[ca-HepTh]
		{\includegraphics[width=0.3\textwidth]{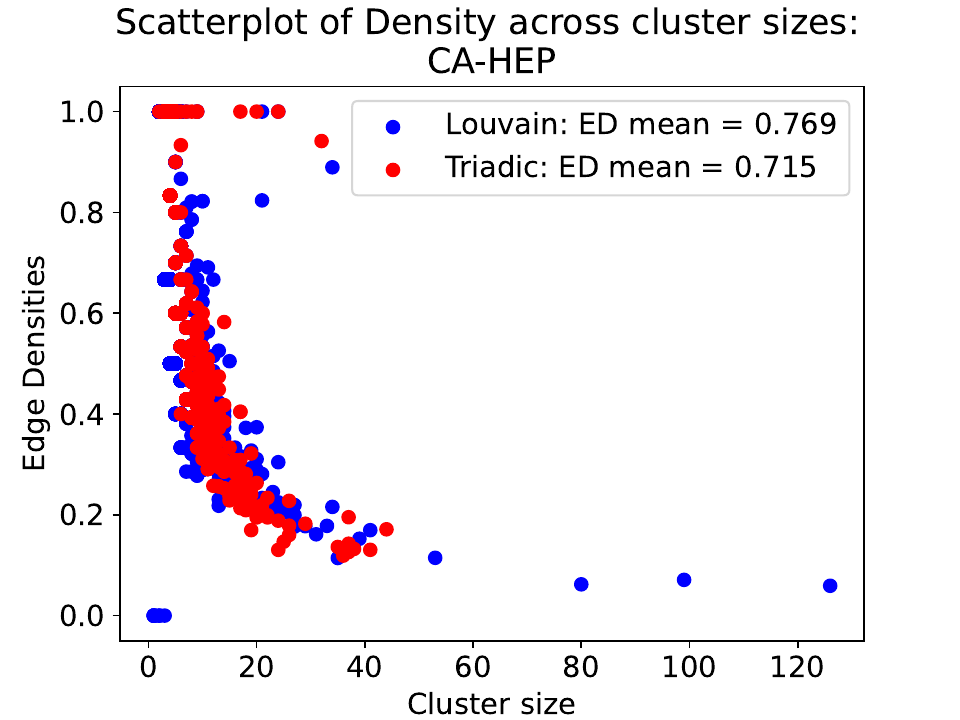}}
	\end{subfigure}
    \begin{subfigure}[ca-CondMat]
		{\includegraphics[width=0.3\textwidth]{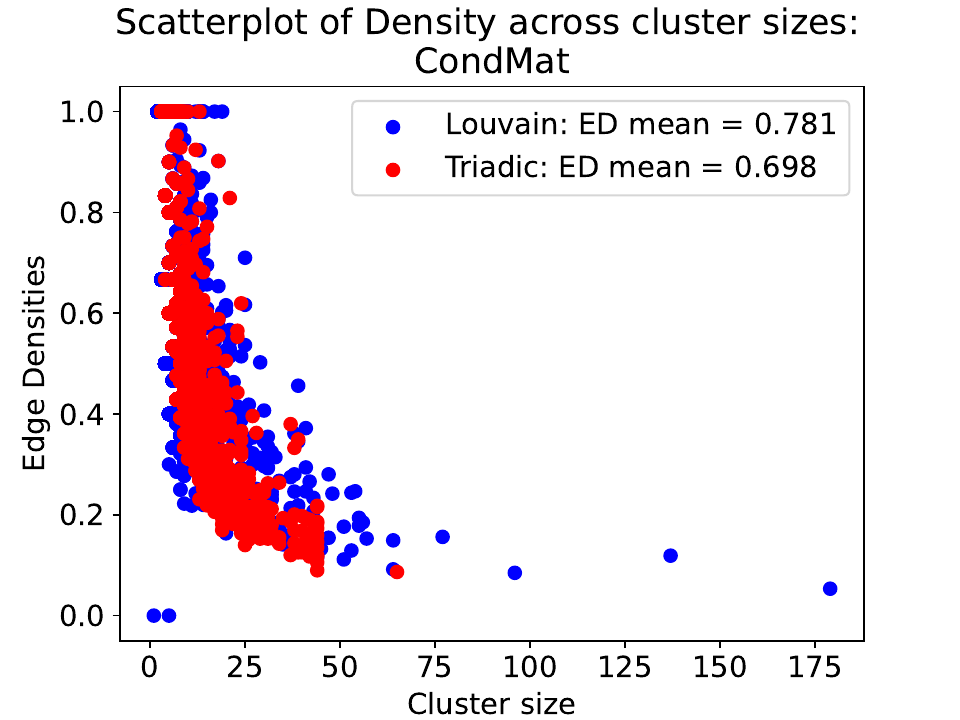}}
	\end{subfigure}	
    \begin{subfigure}[ca-CondMatL]
		{\includegraphics[width=0.3\textwidth]{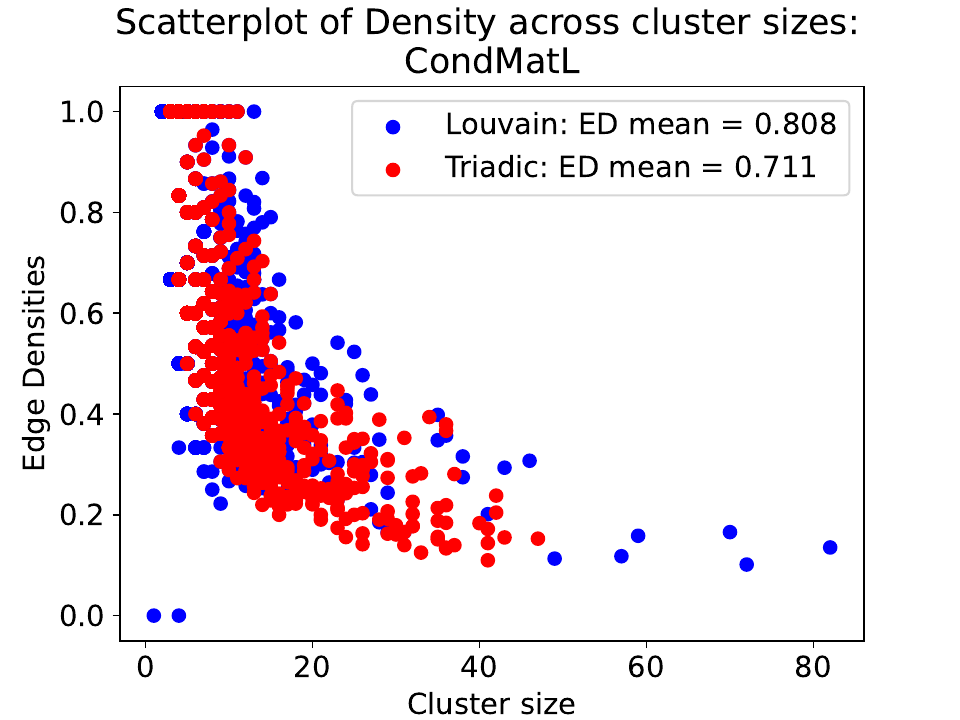}}
	\end{subfigure}
    \begin{subfigure}[cit-DBLP]
		{\includegraphics[width=0.3\textwidth]{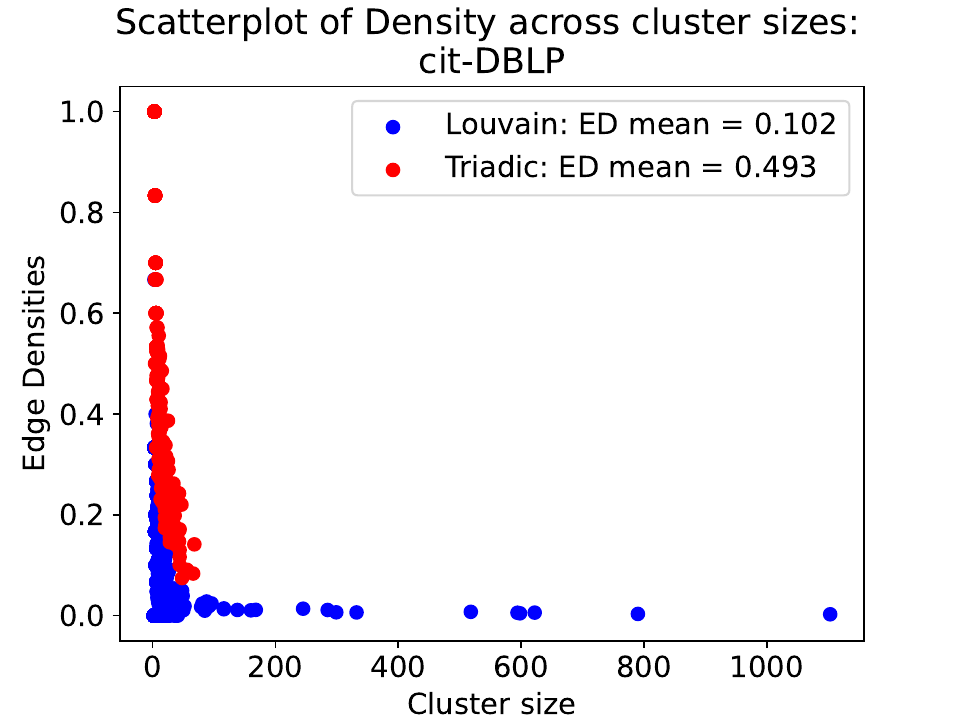}}
	\end{subfigure}
    \captionof{figure}{We give scatter plots for size vs edge density of the spectral triadic clusters and Louvain clusters. In all cases,
	spectral triadic clusters are dense and of moderate size. While many Louvain cluster are also dense, it also creates a few, extremely large, sparse clusters. Overall, the spectral triadic scatterplot
	is above the corresponding Louvain plot, though there is significant overlap.}\label{fig:density}
\end{figure*}

\paragraph{Variation of $\eps$:} 
The algorithm \decompose{} has only one parameter, $\eps$, which
determines the cleaning threshold. The algorithm is fairly stable to changes in this parameter. At higher values of $\eps$, we observe that our clusters have higher edge and triangle density. However, even for values of $\eps$ as high as $0.5$, we preserve a significant number of vertices in the network. For example, in ca-HepTh (\Fig{eps}), at $\eps=0.5$ about 15\% of the vertices are preserved in clusters with mean density over 0.9. In comparison, at $\eps=0.1$, we preserve close to 60\% of the vertices in clusters with a mean density of about $0.7$. 
\begin{figure*}[h!]
	\centering
	\begin{subfigure}[soc-hamsterster]
		{\includegraphics[width=0.3\textwidth]{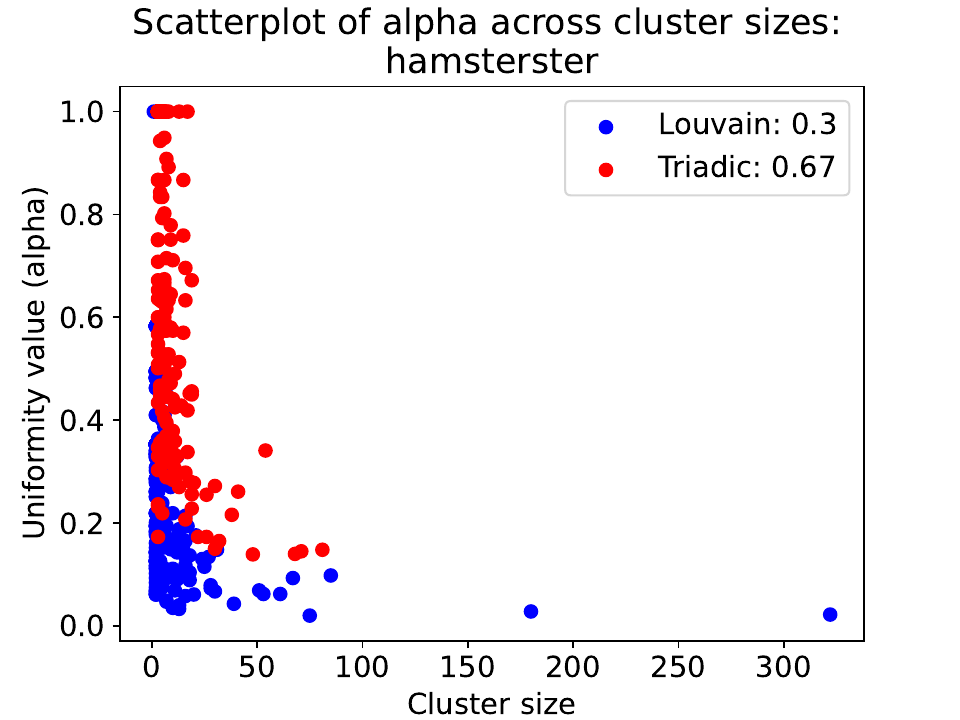}}
	\end{subfigure}
	\begin{subfigure}[socfb-Rice31]
		{\includegraphics[width=0.3\textwidth]{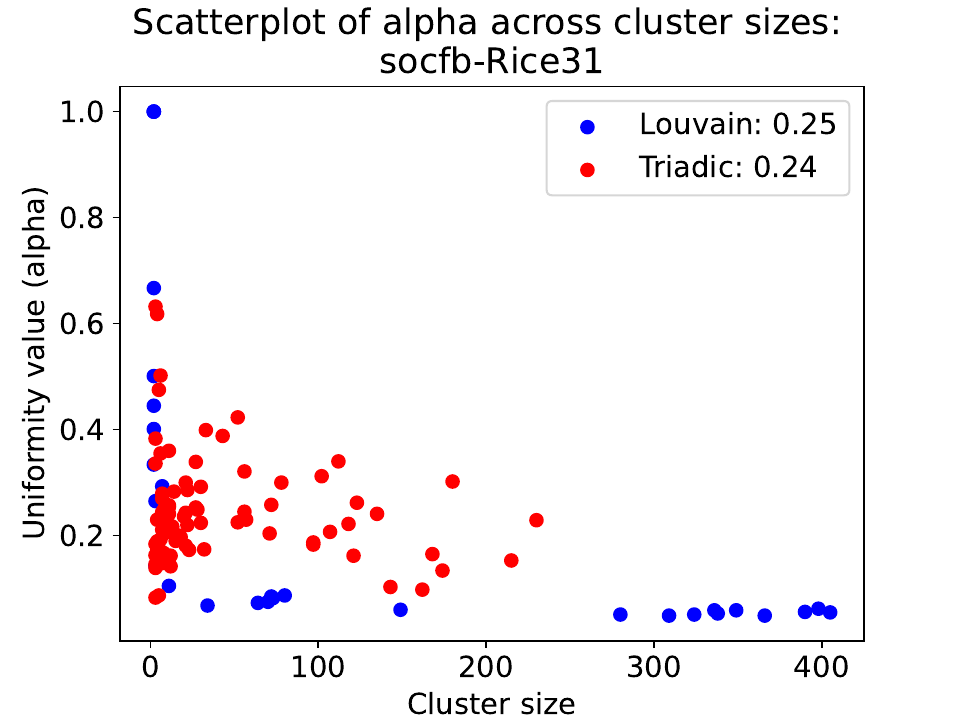}}
	\end{subfigure}
	\begin{subfigure}[ca-HepTh]
		{\includegraphics[width=0.3\textwidth]{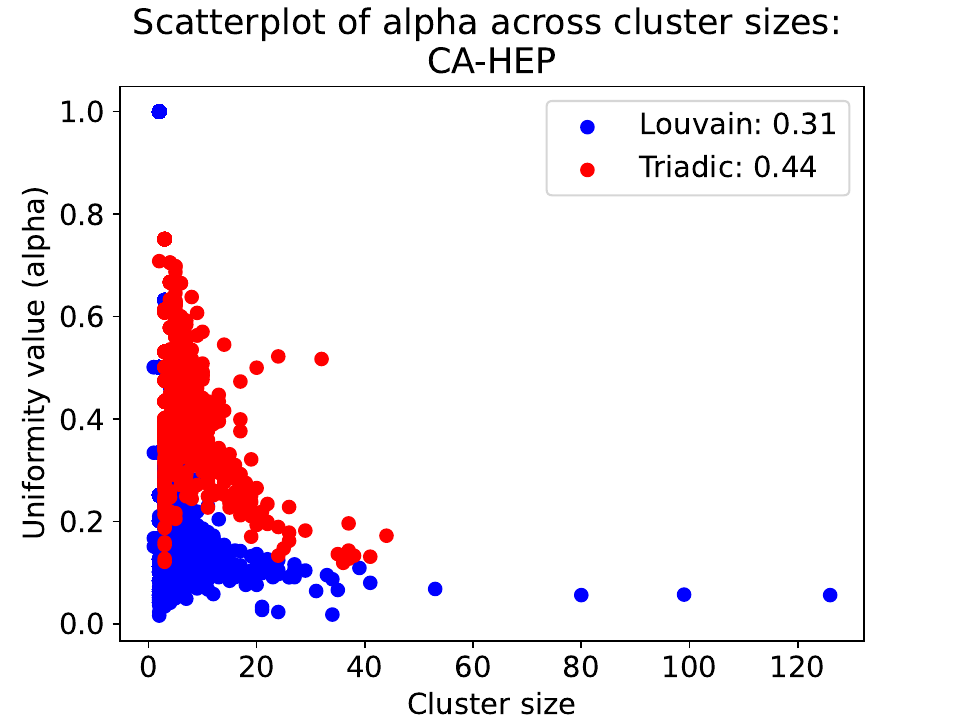}}
	\end{subfigure}

	\captionof{figure}{Uniformity across clusters in the decomposition obtained from various networks as labelled.}\label{fig:unif}
\end{figure*}

\newpage 
\subsection{Examples of semantic significance} \label{sec:semantic}

In \Fig{ground-condmat},
we show graph drawings of two example spectral triadic clusters in a co-authorship network of (over 90K) researchers in Condensed Matter Physics.
The cluster on the left has $16$ vertices and $58$ edges, and has extracted a group of researchers who specialize in optics, ultra fast atoms, and Bose-Einstein condensates.
Notable among them is the 2001 physics Nobel laureate Wolfgang Ketterle. The cluster on the right has $18$ vertices and $55$ edges, and has a group of researchers who all work on nanomaterials;
there are multiple prominent researchers in this cluster, including the 1996 chemistry Nobel laureate Richard Smalley, who discovered buckminsterfullerene.
We stress that the our decomposition found more than a \emph{thousand} such clusters.  
We also note that the Louvain decomposition missed these clusters (end of \Sec{other} and \Fig{comp-cluster}).
\begin{figure*}[h!]
	\centering
	\begin{subfigure}[Condensed Matter Physics: Cluster of researchers working on optics, ultra fast atoms, and Bose-Einstein condensates]
		{\includegraphics[width=0.45\textwidth]{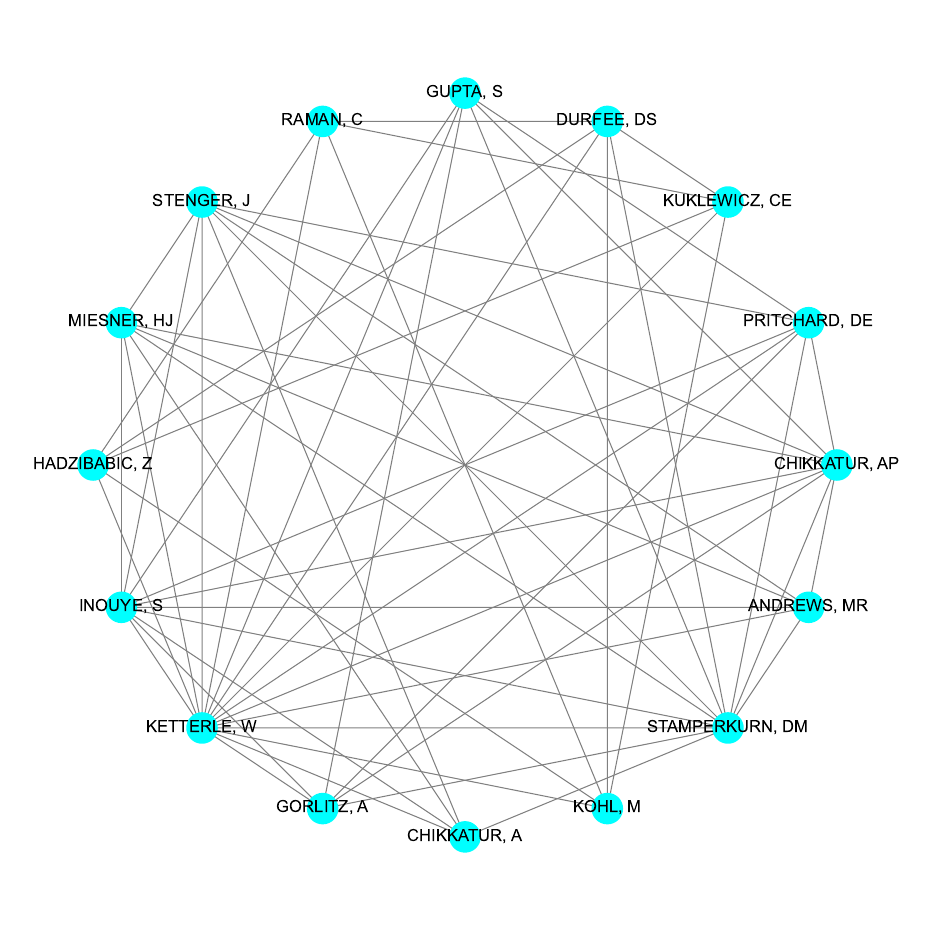}}
	\end{subfigure}
	\begin{subfigure}[Condensed Matter Physics: Cluster of researchers working on graphene, nanomaterials and topological insulators]
		{\includegraphics[width=0.45\textwidth]{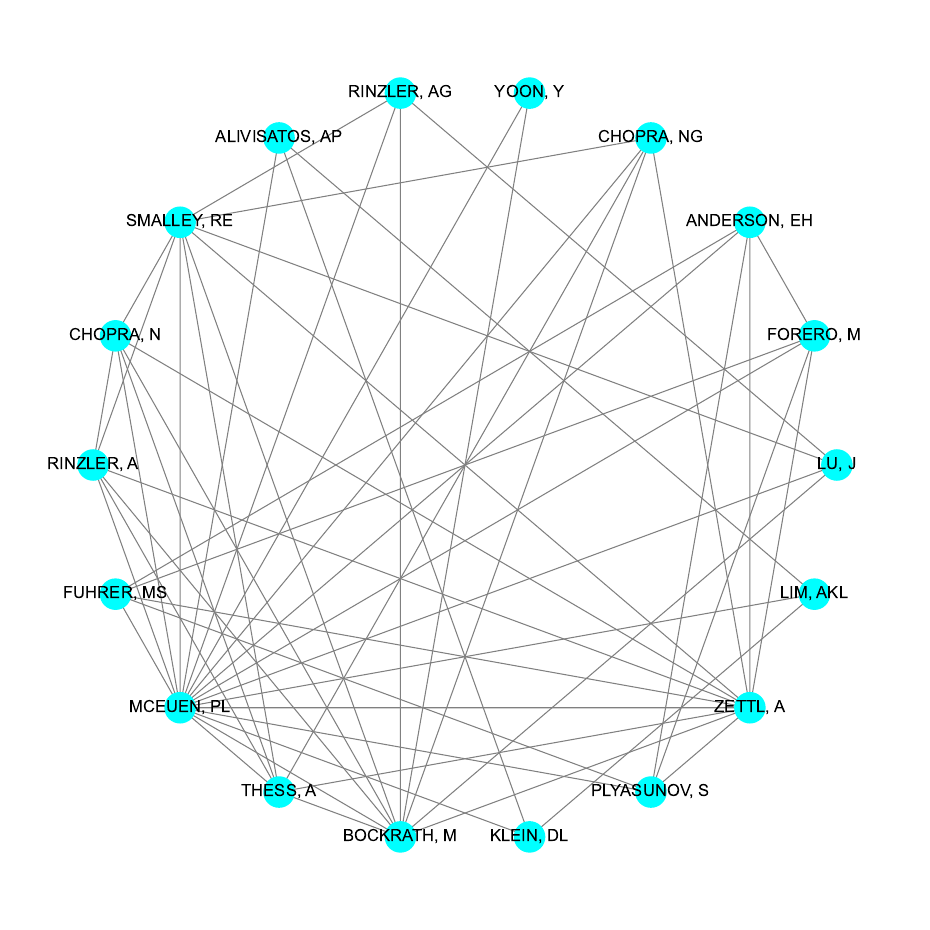}}
	\end{subfigure}
	\captionof{figure}{We show two example clusters from a spectral triadic decomposition
		of coauthorship network of researchers in Condensed Matter Physics \cite{NewmanCondMat99}, a graph with $\tau=0.25$. The left cluster is a set of $16$
		researchers ($58$ edges) working on optics and Bose-Einstein condensates (notably, the cluster has the
		2001 Physics Nobel laureate Wolgang Ketterle). The right cluster has $18$ researchers ($55$ edges)
		working on nanomaterials, including the 1996 Chemistry Nobel laureate Richard Smalley. } \label{fig:ground-condmat}
\end{figure*} As another demonstration, we look at a computer science citation network. Similar extracted clusters of research papers and articles extracted from the DBLP citation network can be seen in \Fig{labelled-dblp}. In this case, one spectral triadic cluster is a group of papers on error correcting/detecting codes, while the other is a cluster of logic program and recursive queries papers. It is surprising how well the spectral triadic decomposition finds fine-grained
structure in networks, based on just the spectral transitivity. This aspect
highlights the practical relevance of spectral theorems that decompose graphs into 
many blocks, rather that the classic Cheeger-type theorems that one produce two blocks.

\begin{figure*}[h!]
	\centering
	\begin{subfigure}[DBLP: Cluster of papers on error correcting codes ]
		{\includegraphics[width=0.45\textwidth]{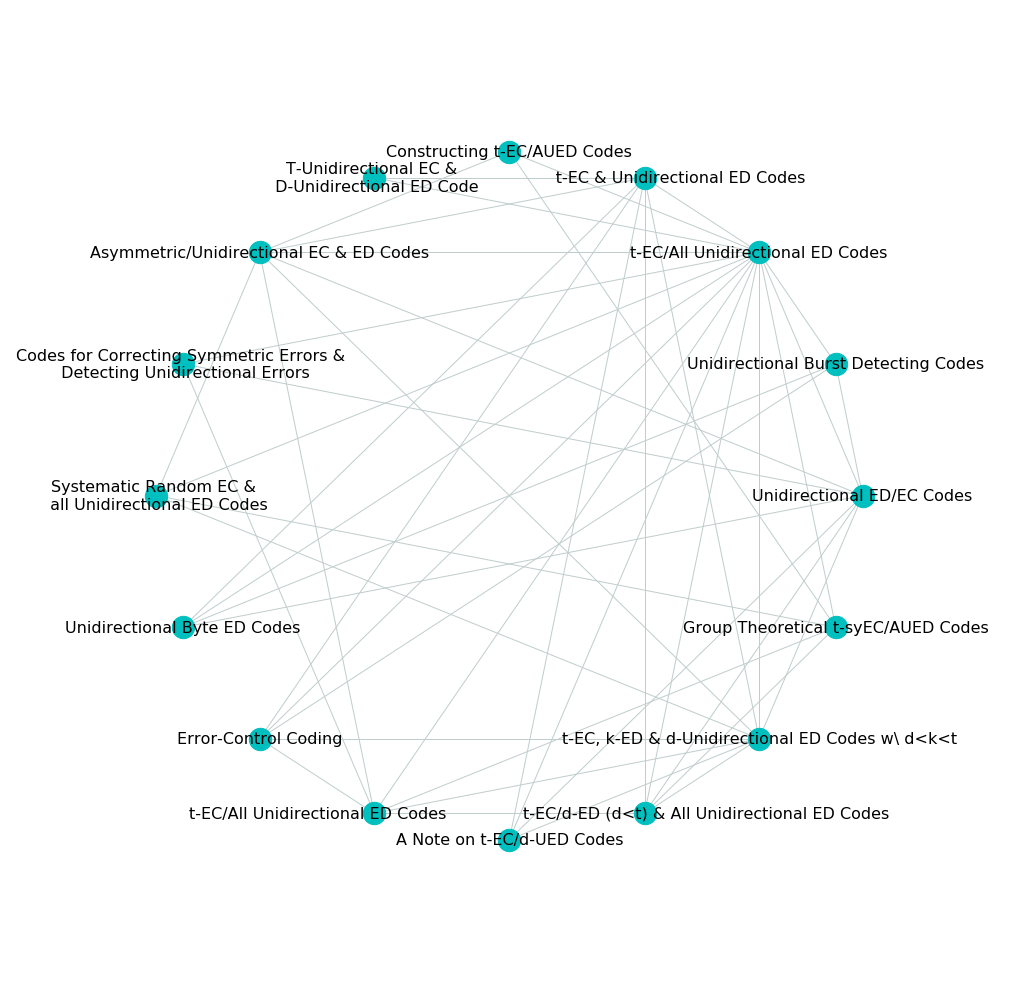}}
	\end{subfigure}
	\begin{subfigure}[DBLP: Cluster of papers on logic programs and recursive queries]
		{\includegraphics[width=0.5\textwidth]{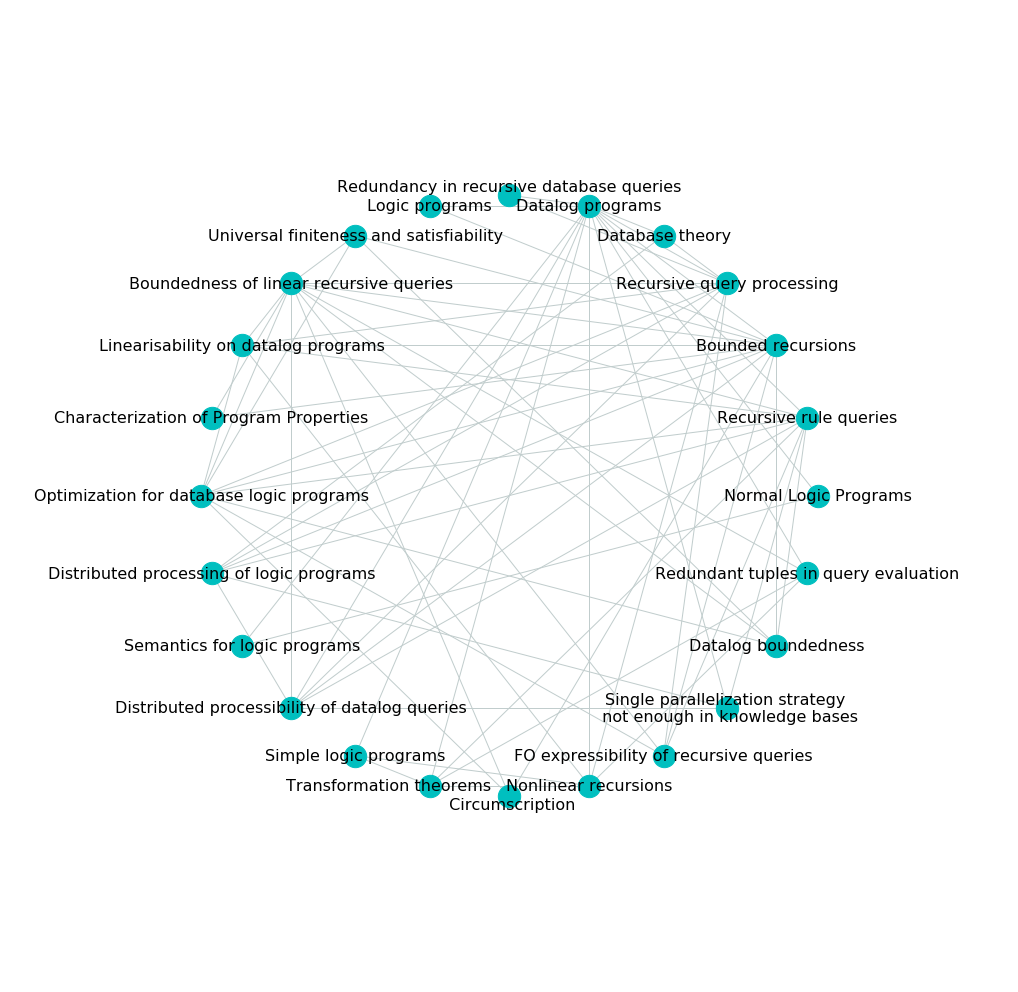}}
	\end{subfigure}
	\captionof{figure}{We show example clusters from a spectral triadic decomposition of a DBLP
		citation network, involving papers in Computer Science ~\cite{ground}. For ease of viewing, we label each vertex with relevant phrases from
		the paper title. The left cluster involves $16$ papers ($47$ edges)
		on the topic of error correcting codes. The right cluster of $24$ papers ($69$ edges) are all on the topic of
		logic programs and recursive queries, from database theory. Observe the tight synergy of topic among
		the vertices in a cluster; our procedure found \emph{thousands} of such clusters.} \label{fig:labelled-dblp}
\end{figure*}

\section{Comparisons with other methods} \label{sec:other}

	In this section, we contrast spectral triadic decompositions with some common community detection algorithms.
    In addition to the real-world datasets discussed earlier, we also experiment with simple stochastic blockmodel (SBM) graphs with a large number of small, dense clusters.
    We run the following methods.
	\begin{asparaitem}
		\item The Louvain algorithm~\cite{Louvain}: This is a classic, fast, community detection algorithm based on a heuristic that maximizes modularity.
            We use an optimized implementation of the authors~\cite{LouvainImp}. We note that results can change by doing fewer outer iterations, but we did not see
            noticeable improvement.
        \item Infomap~\cite{infomap}: This procedures uses a similar heuristic as Louvain but optimizes a different objective called the map equation.
            We use the infomap 2.7.1 package on PyPI~\cite{InfomapPYPI}.
        \item Label propagation~\cite{LabelProp}: This is a procedure analogous to belief propagation, where cluster labels are sent
            along edges. The labels diffuse in the network and converge when a vertex shares its label with many neighbors. 
            We use the community detection library (cdlib) python package~\cite{cdlib}.
        \item $k$-way Spectral cuts~\cite{kway}: This is a classic spectral algorithm to cluster graphs using the Laplacian eigenvectors.
            We use the scikit learn function \texttt{SpectralClustering} from the \texttt{cluster} module ~\cite{kway}.
	\end{asparaitem}
    We refer to our algorithm as ``Triadic" for brevity.

    \medskip

    We summarize our main findings below.
	\begin{asparaitem}
    \item We consider Stochastic Block Model settings, with a simple and clear community structure ($50$ dense communities of $20$ vertices). We note that this is 
        not the standard sparse setting studied in the SBM literature. On the other hand, it is used for network models~\cite{SeKoPi12}. Also, it gives
        a simple evaluation method with a prescribed ``true" community structure. Depending on the parameters,
            Infomap and Label Propagation give trivial results where all clusters are singletons. Louvain gives reasonable results, but makes many errors.
            Only Spectral Triadic Decompositions and spectral $k$-way clustering give perfect results. \Sec{SBMS}
    \item We construct SBMs with a few outliers. In this case, spectral $k$-way clustering is sensitive to the cluster number. When it is given
        an incorrect parameter, it returns one extremely large cluster. All other methods serious errors: either one cluster with too many vertices,
        or many singleton clusters. In all the above cases, only Spectral Triadic Decompositions give correct recovery. \Sec{SBMS}
        \item We note that our SBM settings are distinct from the standard regime of a constant number of sparse blocks. We do not expect our performance to be as strikingly good/near-perfect in that setting. However, that setting also fails the basic premise where our performance guarantees hold (high value of $\stc$), so we do not examine this regime in this article.
    \item On real datasets, spectral $k$-way clustering has extremely poor performance. It tends to create one extremely large cluster with almost all the vertices. \Tab{comp-largest} 
    \item On real datasets, Infomap gives poor results across the board. It is not able to find dense clusters and typically outputs excessively large clusters.
    \item On real datasets, Label Propagation and Louvain sometimes gives good clusters, when measuring the density. But these clusters have poor uniformity,
        implying that they are typically small dense clusters formed by a few high degree vertices. In some instances, label propagation creates
        a single large cluster with poor density. 
        We note that on one of the datasets (ca-cond-matL), these algorithms perform slightly better than Spectral Triadic Decompositions.
    \item As a demonstration, we take the example clusters of~\Fig{ground-condmat}. These were coherent clusters on scientists discovered
        by Spectral Triadic Decompositions on a coauthorship network. We tried to find the closest clusters in decompositions by other algorithms.
        Louvain and Infomap perform somewhat poorly and miss this structure: Louvain clusters are highly disconnected, while  Infomap clusters are typically several hundred vertices large with low overlap with our clusters. An examination of different values of parameters did not yield significantly better results. However, while the Label Propagation has a cluster that is nearly identical
        (but does extremely poorly on SBMs). 
	\end{asparaitem}
	
	\subsection{Stochastic Block Models: Testing for prescribed structure}\label{sec:SBMS}
	We use $\SBM(n_1, n_2, p_1, p_2)$ to denote the following stochastic block model.
    with $n_1$ blocks of size $n_2$ each, with in-probability of edges being $p_1$ and out-probability being $p_2$.  The vertices are partitioned into $n_1$ blocks of $n_2$ vertices. With a block, an edge is put with independent probability $p_1$. For edges between blocks, an edge is put with independent probability $p_2$. 

    For an illustration, we consider the setting $\SBM(50, 20, 0.9, 0.1)$. Observe that the ground truth ``communities" are highly internally connected,
    but there is some fraction of ``noise" that connects different communities/blocks.
	The results for our stochastic block model experiments are summarized in \Fig{SBM1}.
	\begin{figure}

        \begin{minipage}{0.4\textwidth}
			\includegraphics[width=0.8\textwidth]{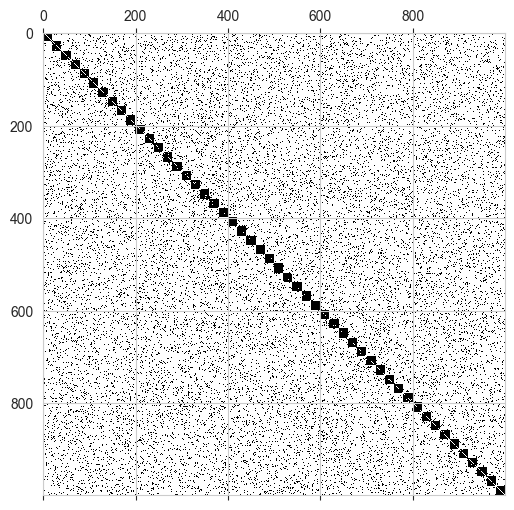}
		\end{minipage}
		\begin{minipage}{0.4\textwidth}
			\centering
			\begin{tabular}{|c||c|}
				\hline
				\multicolumn{2}{c}{$\SBM(50, 20, 0.9, 0.1)$}\\
				\hline
				Algorithm & Clusters (number$\times$size)  \\
				\hline
				Label Propagation &\cellcolor{red!25} $1\times1000$\\
				Louvain & $2\times 20, 8\times40, 3\times60, 4\times 80, 1\times 140$\\
				Infomap &\cellcolor{red!25} $1\times 1000$\\
				$50$-way spectral &\cellcolor{green!25} $20\times 50$ \\
				Triadic &\cellcolor{green!25}$20\times 50$\\
				\hline
			\end{tabular}
			
		\end{minipage}
		\caption{The outcomes of different algorithms on the $\SBM(50,20,0.9,0.1)$. The left image
        has a spyplot where the cluster structure is clearly visible. The tables describe the clusters found by various methods. The green rows
        indicate perfect recovery and the red rows denote a trivial clustering (one cluster with all vertices).}\label{fig:SBM1}
	\end{figure}
   	\begin{table}
   	
   	\begin{tabular}{|c||c|}
   		\hline
   		\multicolumn{2}{c}{SBM with outliers}\\
   		\hline
   		Algorithm & Clusters (number$\times$size)  \\
   		\hline
   		Label Propagation &\cellcolor{red!25} $1\times1000, 50\times 2$\\
   		Louvain & $12\times 22, 12\times 44, 2\times66, 2\times 88$\\
   		Infomap & $38\times 22, 6 \times 44$\\
   		$50$-way spectral &\cellcolor{red!25} $49\times 2, 1002\times 1$ \\
   		$100$-way spectral & $50\times 2, 50\times 20$ \\
   		Triadic &\cellcolor{green!25}$100\times 1, 50\times 20$\\
   		\hline
   	\end{tabular}
   	\begin{tabular}{|c||c|}
   		\hline
   		\multicolumn{2}{c}{SBM with high degree node and outliers}\\
   		\hline
   		Algorithm & Clusters (number$\times$size)  \\
   		\hline
   		Label Propagation &\cellcolor{red!25} $1\times 999, 1\times 101$\\
   		Louvain & $10\times 20, 15\times 40, 3\times60, 1\times 120$\\
   		Infomap & $35\times 20, 6 \times 40, 1\times 59, 1\times 101$\\
   		$50$-way spectral &\cellcolor{red!25} $49\times 2, 1002\times 1$ \\
   		$100$-way spectral & $50\times 2, 50\times 20$ \\
   		Triadic &\cellcolor{green!25}$100\times 1, 50\times 20$\\
   		\hline
   	\end{tabular}
   	\caption{Summary of clusters by different methods in SBMs with outliers. Both have a ground truth of $50$ clusters of size $20$, with $100$ outlier vertices. Again, we see that Spectral Triadic Decompositions
   		perform an exact recovery. Notably, $k$-way spectral clustering fails when give the right value of $k$, and needs $k = 100$ to succeed.} \label{tab:noisesbmtab}
   \end{table}
Both Infomap and Label Propagation simply cluster all vertices into a single set in both cases. Louvain gives better results
by does create some large clusters of size 80 and 140. Only Triadic and 50-way spectral get perfect recovery.

But we observe that $k$-way spectral is highly sensitive to the value $k$, which is unknown (and extremely large) in practice.
We consider another experiment, where we add some outlier noise to the SBM. We add 100 outlier vertices to the SBM in two different ways.
In the first case, we randomly pair these vertices into 50 disjoint edges. Then, each pair is connected to a random vertex.
(This is called ``SBM with outliers".) In the second case, one of these outliers are made into a ``celebrity", where
they connect to 100 vertices (SBM with high degree node and outliers).

The results are in \Tab{noisesbmtab}. In all case, Triadic gets perfect recovery. When we run spectral cluster with $k=50$,
the results are completely erroneous. One needs to set a larger $k = 100$ to get anything non-trivial. 
All other methods make various errors, either creating a cluster that is too large,
or creating too many tiny clusters.

\subsection{Real Graphs}\label{sec:realsec}
We now compare the results on real graphs. For convenience, we focus on 4 of the datasets, though we see consistent results. 
We present specific shortcomings of all the other methods. 

\paragraph{Excessively large clusters:} Barring Louvain (and Triadic), all other methods tend to create one large cluster
of extremely low density. Such a cluster would not be considered a community by any definition. We present the data in \Tab{comp-largest}.
Each row corresponds to a method. Each entry gives the number and fraction of vertices in the largest cluster created
that method, for a given dataset. We also give the edge density of the cluster. 

We observe the absymal performance of $k$-way spectral clustering, which often creates a cluster with more than 90\% of all vertices. 
The edge densities are less than $10^{-3}$ in almost all cases.
Label Propagation and Infomap also create large clusters with low density, though these values are better than $k$-way spectral clustering.
For socfb-Rice31 graph, Label Propagation creates a cluster with 99.9\% of the vertices having edge density $0.02$.
For the soc-hamsterster graph, Infomap creates a cluster with 56\% of the vertices, and edge density $0.01$. 
In contrast, both Triadic creates hundreds of clusters of size at most hundreds, with edge densities at least an order of magnitude
higher. Moreover, the largest triadic clusters cover between 0.3\% and 5\% of the graphs (\Tab{comp-largest}).

Only Louvain has a reasonable performance across all datasets. But as we showed in~\Sec{details-exp}, the edge densities are
lower than that of Spectral Triadic Decompositions. Moreover, from \Tab{comp-largest} we can see that Louvain tends to aggregate into larger clusters.

\begin{table}
	\centering
	\begin{tabular}{|c||c|c||c|c||c|c||c|c|}
		\hline
		\multirow{2}{*}{Algorithm} &\multicolumn{2}{c||}{soc-hamsterster} &\multicolumn{2}{c||}{socfb-Rice31} &\multicolumn{2}{c||}{ca-cond-matL} &\multicolumn{2}{c|}{ca-cond-mat}\\ \cline{2-9}
		& Max(\%$|V|$) & Dens. &Max(\%$|V|$) & Dens. & Max(\%$|V|$) & Dens. & Max(\%$|V|$) & Dens. \\
		\hline
		20-way spec&2197 (91.6\%)&7$\cdot10^{-3}$&1943 (47.5\%)&4$\cdot10^{-2}$&14854 (91.3\%)&7$\cdot10^{-4}$&22467 (97.1\%)&4$\cdot10^{-4}$\\
		60-way spec&2142 (80.7\%)&7$\cdot10^{-3}$&2120 (51.9\%)&4$\cdot 10^{-2}$&14535 (89.4\%)&8$\cdot10^{-4}$&22132 (95.7\%)&4$\cdot10^{-4}$\\
		100-way spec&2060 (77.8\%)&8$\cdot10^{-3}$&2257 (55.2\%)&4$\cdot10^{-2}$&14261 (87.7\%)&8$\cdot10^{-4}$&21924 (94.8\%)&4$\cdot10^{-4}$\\ \hline
		Label Prop & 1378 (56.8\%) & 0.013 & 4083 (99.9\%) & 0.02 & 79 (0.5\%) & 0.118 & 329 (1.4\%) & 0.03 \\ \hline
		Infomap & 1989 (82.0\%) & 0.008 & 432 (10.6\%) & 0.076 & 1522 (9.4\%) & 0.003 & 6684 (28.9\%) & 0.001 \\ \hline
		Louvain & 322 (13.3\%) & 0.046 & 405 (9.9\%) & 0.103 & 82 (0.5\%) & 0.135 & 179 (1.1\%) & 0.054 \\ \hline
		Triadic & 81 (3.3\%) & 0.148 & 230 (5.6\%) & 0.228 & 47 (0.3\%) & 0.153 & 65 (0.3\%) & 0.087 \\
		\hline
	\end{tabular}
\caption{Sizes of the largest clusters extracted by each method, and their corresponding edge densities. Max refers to the size of the largest cluster. 
    In brackets, we show the size as a percentage of the whole graph. It is immediate that other methods give larger, sparser clusters that Triadic.} \label{tab:comp-largest}
\end{table}

\begin{figure}[h!]
\centering
\begin{subfigure}[Louvain: overlap of 3]
	{\includegraphics[width=0.30\textwidth]{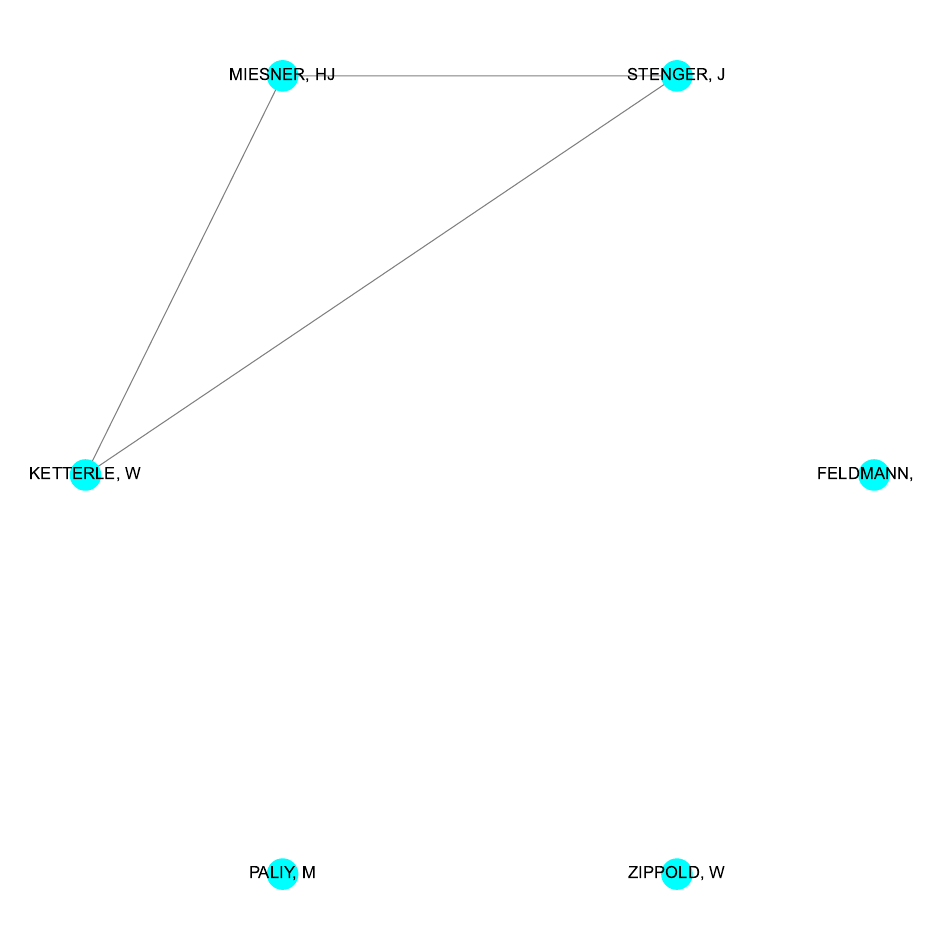}}
\end{subfigure}
\begin{subfigure}[Infomap: overlap of 7]
	{\includegraphics[width=0.30\textwidth]{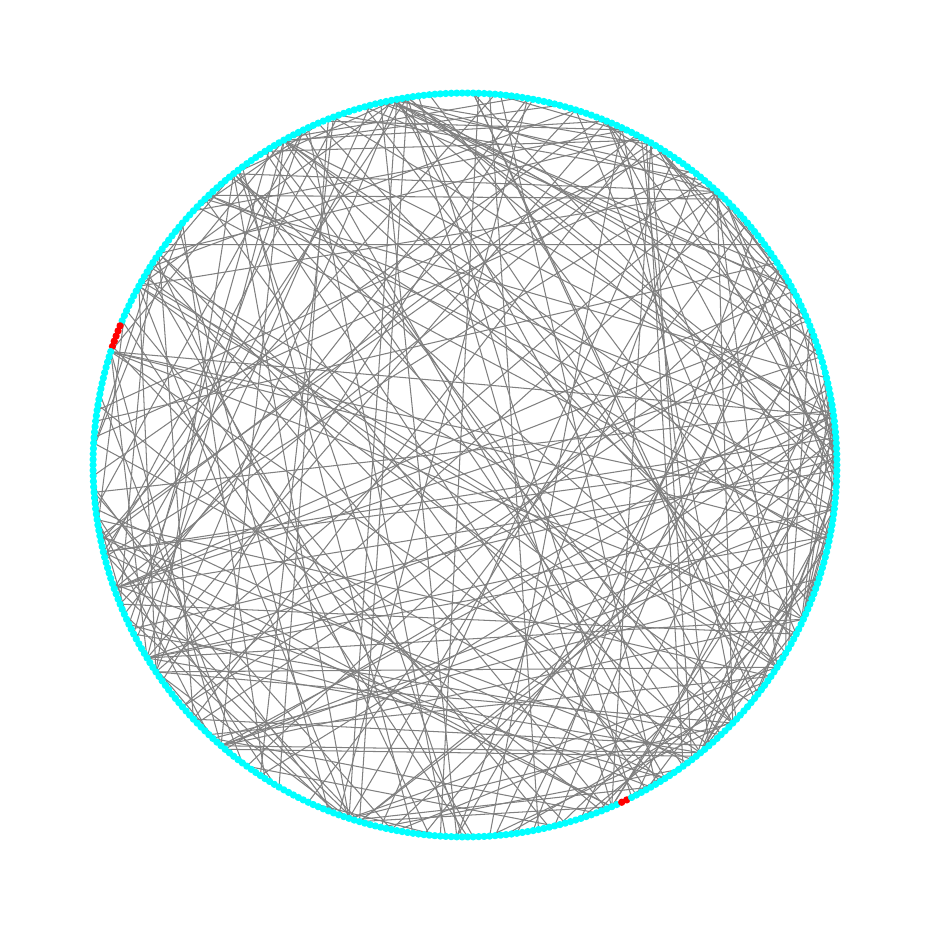}}
\end{subfigure}
\begin{subfigure}[Label propagation: overlap of 16]
	{\includegraphics[width=0.30\textwidth]{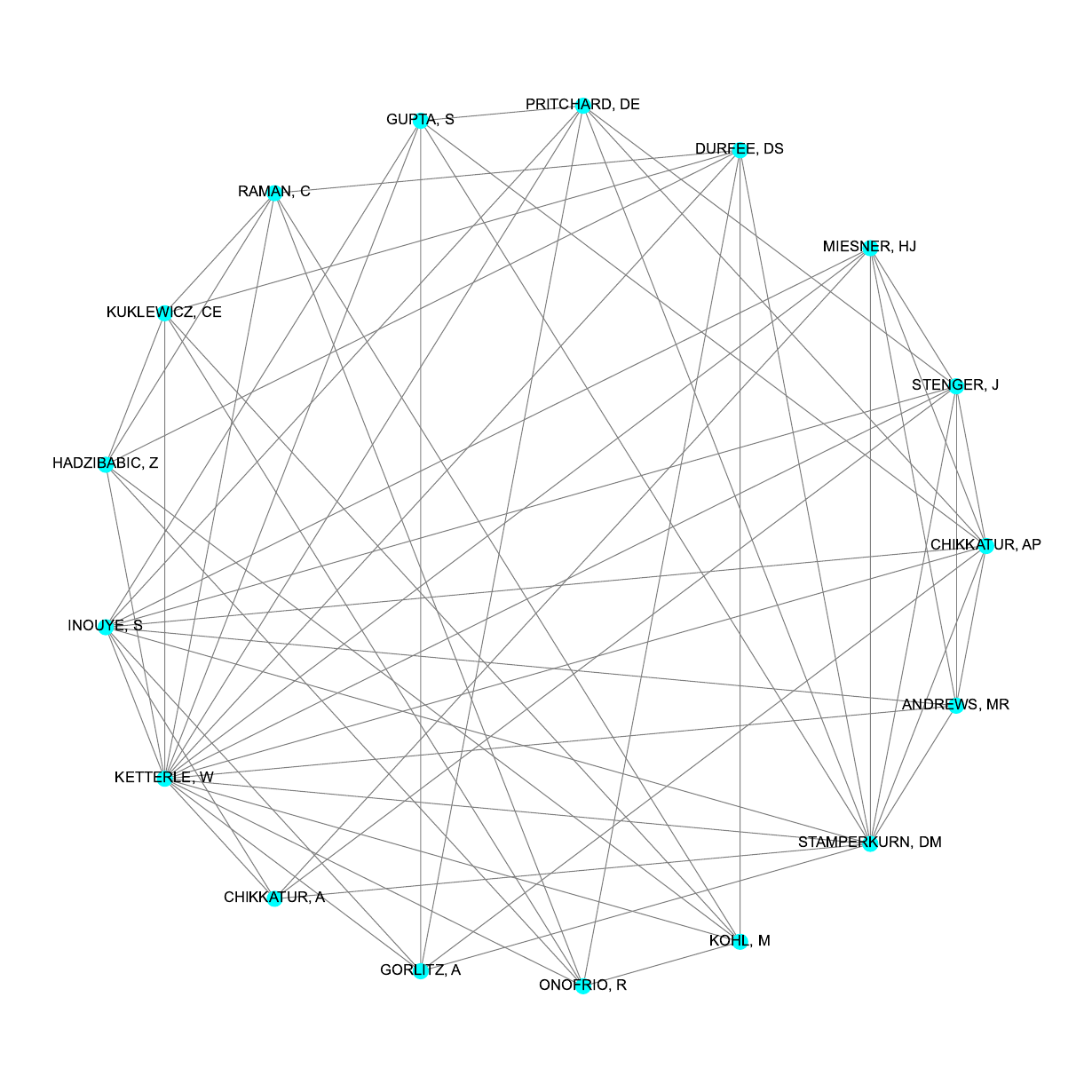}}
\end{subfigure}
\begin{subfigure}[Louvain: overlap of 3]
	{\includegraphics[width=0.30\textwidth]{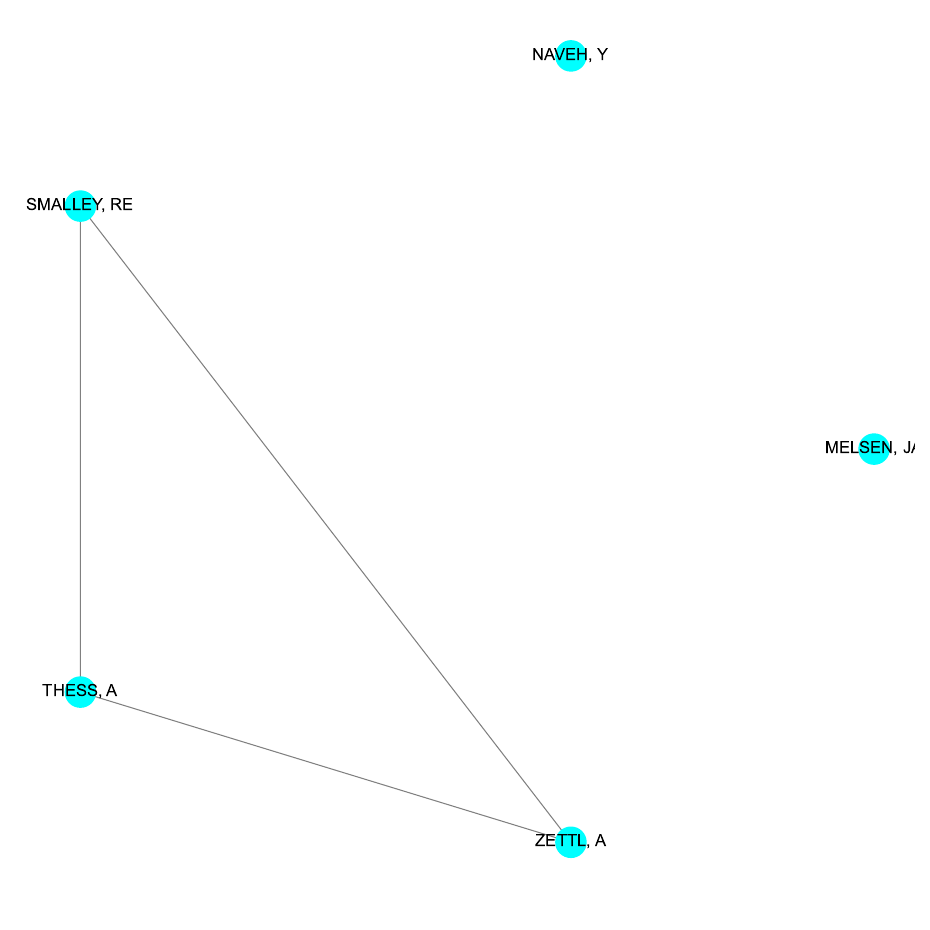}}
\end{subfigure}
\begin{subfigure}[Infomap: overlap of 2]
	{\includegraphics[width=0.30\textwidth]{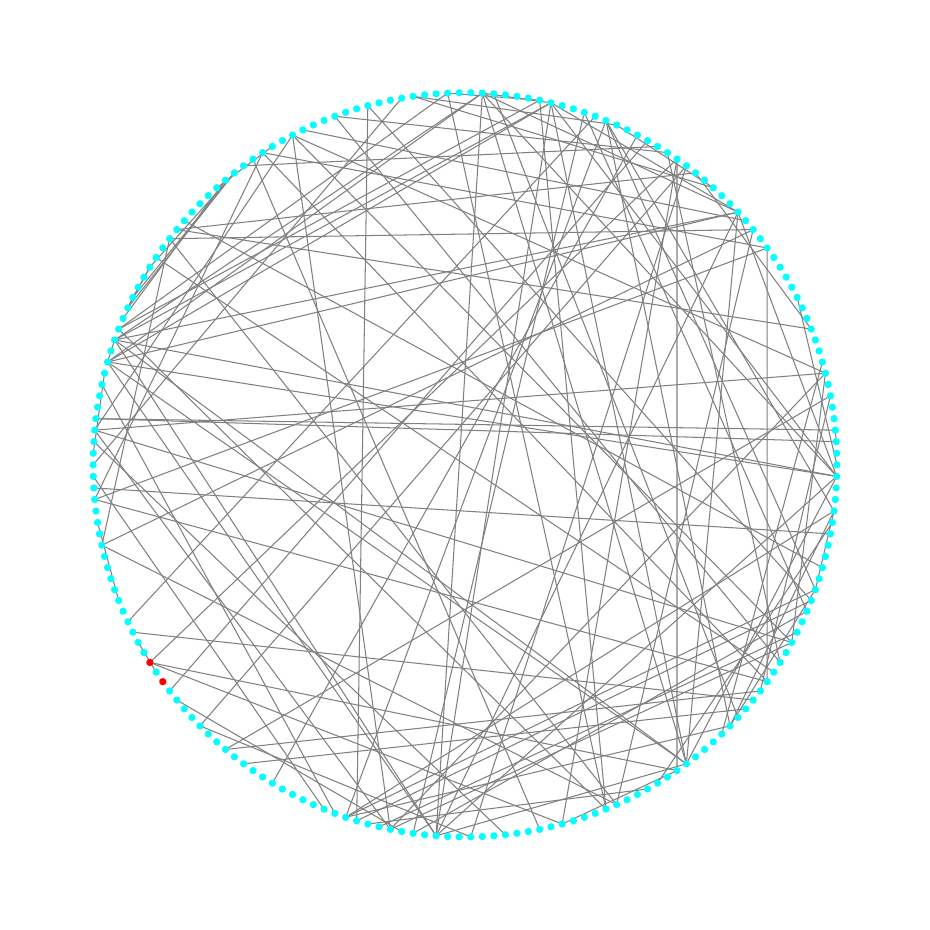}}
\end{subfigure}
\begin{subfigure}[Label propagation: overlap of 15]
	{\includegraphics[width=0.30\textwidth]{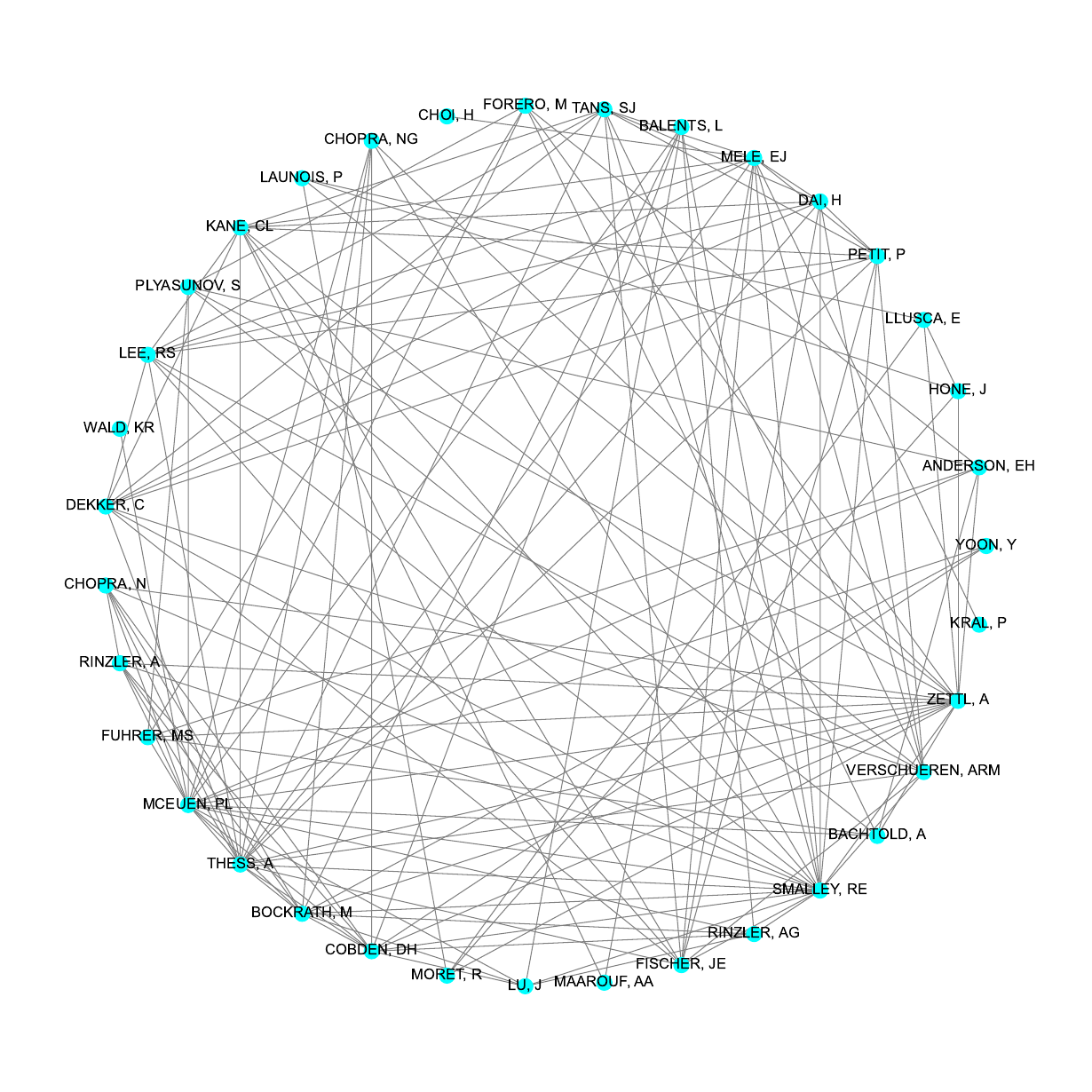}}
    \caption{We consider the Triadic clusters given in \Fig{ground-condmat}, and look for the closest among
    the Louvain, Infomap, and Label Propagation clusters. We observe that Louvain and Infomap completely miss the structure. Label
    Propagation has an extremely close match.} \label{fig:comp-cluster}
\end{subfigure}

\end{figure}

\paragraph{Comparison on labeled clusters:} We take the two clusters from \Fig{ground-condmat} found by Triadic, on a ca-condMatL.
Recall that these were both semantically meaningful, with clusters of researchers on optics and nanomaterials respectively.
For each cluster, we look for the best match among the clusters generated by (say) Louvain. We simply look for the Louvain
cluster with the highest overlap. This process is repeated for the other methods. 

The closest clusters are shown in \Fig{comp-cluster}. Louvain is not able to extract this semantically meaningful cluster.
The closest cluster is much smaller, sparse and disconnected, as noted in previous work~\cite{Leiden}. An examination in different parameters/settings or even even higher levels in the Louvain tree did not yield significantly better results.
Thus, this demonstrates the better performance on Triadic on discovering ground truth. Infomap also fails, but in the 
opposite direction. It finds an excessively large cluster with poor overlap. Label Propagation performs the best
and is able to get extremely similar clusters.

	\bibliographystyle{plain}
	\bibliography{triangle-dense}

\begin{thebibliography}{10}

\bibitem{LouvainImp}
C++ code for louvain algorithm.
\newblock \url{https://sourceforge.net/projects/louvain/}.

\bibitem{cdlib}
Code for label propagation.
\newblock
  \url{https://cdlib.readthedocs.io/en/latest/reference/cd_algorithms/algs/cdlib.algorithms.label_propagation.html}.

\bibitem{InfomapPYPI}
Pypi package: Infomap.
\newblock \url{https://pypi.org/project/infomap/}.

\bibitem{ACL06}
Reid Andersen, Fan Chung, and Kevin Lang.
\newblock Local graph partitioning using pagerank vectors.
\newblock pages 475--486, 2006.

\bibitem{AnChLa06}
Reid Andersen, Fan Chung, and Kevin Lang.
\newblock Local graph partitioning using pagerank vectors.
\newblock {\em Proc. of Foundations of Comp. Sci (FOCS)}, pages 475--486, 2006.

\bibitem{BeGlLe16}
A.~Benson, D.~F. Gleich, and J.~Leskovec.
\newblock Higher-order organization of complex networks.
\newblock {\em Science}, 353(6295):163--166, 2016.

\bibitem{Louvain}
Vincent~D Blondel, Jean-Loup Guillaume, Renaud Lambiotte, and Etienne Lefebvre.
\newblock Fast unfolding of communities in large networks.
\newblock {\em Journal of Statistical Mechanics: Theory and Experiment},
  2008(10):P10008, oct 2008.

\bibitem{Burt04}
Ronald~S. Burt.
\newblock Structural holes and good ideas.
\newblock {\em American Journal of Sociology}, 110(2):349--399, 2004.

\bibitem{CF06}
Deepayan Chakrabarti and Christos Faloutsos.
\newblock Graph mining: Laws, generators, and algorithms.
\newblock {\em ACM Comput. Surv.}, 38(1), Jun 2006.

\bibitem{CSX12}
Yudong Chen, Sujay Sanghavi, and Huan Xu.
\newblock Clustering sparse graphs.
\newblock In F.~Pereira, C.J. Burges, L.~Bottou, and K.Q. Weinberger, editors,
  {\em Advances in Neural Information Processing Systems}, volume~25, 2012.

\bibitem{Chung:1997}
F.~R.~K. Chung.
\newblock {\em Spectral Graph Theory}.
\newblock American Mathematical Society, 1997.

\bibitem{Faust}
Katherine Faust.
\newblock Comparing social networks: Size, density, and local structure.
\newblock {\em Metodoloski zvezki}, 3:185--216, 07 2006.

\bibitem{FORTUNATO201075}
Santo Fortunato.
\newblock Community detection in graphs.
\newblock {\em Physics Reports}, 486(3):75--174, 2010.

\bibitem{FB07}
Santo Fortunato and Marc Barthélemy.
\newblock Resolution limit in community detection.
\newblock {\em Proceedings of the National Academy of Sciences}, 104(1):36--41,
  2007.

\bibitem{GiNe02}
M.~Girvan and M.~Newman.
\newblock Community structure in social and biological networks.
\newblock {\em Proceedings of the National Academy of Sciences},
  99(12):7821--7826, 2002.

\bibitem{GlSe12}
D.~Gleich and C.~Seshadhri.
\newblock Vertex neighborhoods, low conductance cuts, and good seeds for local
  community methods.
\newblock pages 597--605, 2012.

\bibitem{Gr83}
M.~Granovetter.
\newblock The strength of weak ties: A network theory revisited.
\newblock {\em Sociological Theory}, 1:201--233, 1983.

\bibitem{GRS}
Rishi Gupta, Tim Roughgarden, and C.~Seshadhri.
\newblock Decompositions of triangle-dense graphs.
\newblock {\em Innovations in Theoretical Computer Science}, pages 471--482,
  2014.

\bibitem{HoLe70}
P.~Holland and S.~Leinhardt.
\newblock A method for detecting structure in sociometric data.
\newblock {\em American Journal of Sociology}, 76:492--513, 1970.

\bibitem{HL76}
Paul~W. Holland and Samuel Leinhardt.
\newblock Local structure in social networks.
\newblock {\em Sociological Methodology}, 7:1--45, 1976.

\bibitem{JS89}
Mark Jerrum and Alistair Sinclair.
\newblock Approximating the permanent.
\newblock {\em SIAM Journal on Computing}, 18(6):1149--1178, 1989.

\bibitem{Kl00}
J.~Kleinberg.
\newblock Navigation in a small world.
\newblock {\em Nature}, 406(6798), 2000.

\bibitem{LS88}
Gregory~F. Lawler and Alan~D. Sokal.
\newblock Bounds on the $l^2$ spectrum for markov chains and markov processes:
  A generalization of cheeger's inequality.
\newblock {\em Transactions of the American Mathematical Society},
  309(2):557--580, 1988.

\bibitem{LeGh+14}
James~R. Lee, Shayan~Oveis Gharan, and Luca Trevisan.
\newblock Multiway spectral partitioning and higher-order cheeger inequalities.
\newblock {\em J. ACM}, 61(6):1--30, 2014.

\bibitem{LeLaDa08}
J.~Leskovec, K.~J. Lang, A.~Dasgupta, and M.~W. Mahoney.
\newblock Community structure in large networks: Natural cluster sizes and the
  absence of large well-defined clusters.
\newblock {\em Internet Mathematics}, 6(1):29--123, 2009.

\bibitem{snapnets}
Jure Leskovec and Andrej Krevl.
\newblock {SNAP Datasets}: {Stanford} large network dataset collection.
\newblock \url{http://snap.stanford.edu/data}, June 2014.

\bibitem{LLDM08}
Jure Leskovec, Kevin~J. Lang, Anirban Dasgupta, and Michael~W. Mahoney.
\newblock Statistical properties of community structure in large social and
  information networks.
\newblock In {\em Proceedings of the 17th International Conference on World
  Wide Web}, WWW '08, page 695–704, 2008.

\bibitem{Mi67}
S.~Milgram.
\newblock The small world problem.
\newblock {\em Psychology Today}, 1(1):60--67, 1967.

\bibitem{NewmanCondMat99}
M.~E.~J. Newman.
\newblock The structure of scientific collaboration networks.
\newblock {\em Proceedings of the National Academy of Sciences},
  98(2):404--409, 2001.

\bibitem{Newman03}
M.~E.~J. Newman.
\newblock Properties of highly clustered networks.
\newblock {\em Phys. Rev. E}, 68:026121, Aug 2003.

\bibitem{New06}
M.~E.~J. Newman.
\newblock Modularity and community structure in networks.
\newblock {\em Proceedings of the National Academy of Sciences},
  103:8577--8582, 2006.

\bibitem{NR16}
M.~E.~J. Newman and Gesine Reinert.
\newblock Estimating the number of communities in a network.
\newblock {\em Phys. Rev. Lett.}, 117:078301, Aug 2016.

\bibitem{kway}
Andrew~Y. Ng, Michael~I. Jordan, and Yair Weiss.
\newblock On spectral clustering: Analysis and an algorithm.
\newblock In {\em Proceedings of the 14th International Conference on Neural
  Information Processing Systems: Natural and Synthetic}, NIPS'01, page
  849–856, 2001.

\bibitem{PSZ15}
Richard Peng, He~Sun, and Luca Zanetti.
\newblock Partitioning well-clustered graphs: Spectral clustering works!
\newblock In {\em Conference on Learning Theory (COLT)}, volume~40 of {\em
  Proceedings of Machine Learning Research}, pages 1423--1455, 2015.

\bibitem{Radicchi}
Filippo Radicchi, Claudio Castellano, Federico Cecconi, Vittorio Loreto, and
  Domenico Parisi.
\newblock Defining and identifying communities in networks.
\newblock {\em Proceedings of the National Academy of Sciences},
  101(9):2658--2663, 2004.

\bibitem{LabelProp}
Usha~Nandini Raghavan, R\'eka Albert, and Soundar Kumara.
\newblock Near linear time algorithm to detect community structures in
  large-scale networks.
\newblock {\em Phys. Rev. E}, 76:036106, Sep 2007.

\bibitem{nr}
Ryan~A. Rossi and Nesreen~K. Ahmed.
\newblock The network data repository with interactive graph analytics and
  visualization.
\newblock {\em AAAI}, 2015.

\bibitem{infomap}
Martin Rosvall and Carl~T. Bergstrom.
\newblock Maps of random walks on complex networks reveal community structure.
\newblock {\em Proceedings of the National Academy of Sciences},
  105(4):1118--1123, 2008.

\bibitem{SaSePi+15}
A.~Erdem Sariyuce, C.~Seshadhri, A.~Pinar, and U.~Catalyurek.
\newblock Finding the hierarchy of dense subgraphs using nucleus
  decompositions.
\newblock In {\em World Wide Web (WWW)}, pages 927--937, 2015.

\bibitem{SeKoPi12}
C.~Seshadhri, Tamara~G. Kolda, and Ali Pinar.
\newblock Community structure and scale-free collections of {{E}rdos-{R}enyi}
  graphs.
\newblock {\em Physical Review E}, 85:056109, 2012.

\bibitem{SpielmanSAGT}
Daniel~A. Spielman.
\newblock {\em Spectral and Algebraic Graph Theory}.
\newblock \url{http://cs-www.cs.yale.edu/homes/spielman/sagt/}.

\bibitem{SpielmanTutorial}
Daniel~A. Spielman.
\newblock Spectral graph theory and its applications.
\newblock In {\em IEEE Symposium on Foundations of Computer Science (FOCS)},
  pages 29--38, 2007.

\bibitem{ST08}
Daniel~A. Spielman and Shang-Hua Teng.
\newblock A local clustering algorithm for massive graphs and its application
  to nearly linear time graph partitioning.
\newblock {\em SIAM Journal on Computing}, 42(1):1--26, 2013.

\bibitem{aminer}
Jie Tang, Jing Zhang, Limin Yao, Juanzi Li, Li~Zhang, and Zhong Su.
\newblock Arnetminer: Extraction and mining of academic social networks.
\newblock {\em SIGKDD Conference on Knowledge Discovery and Data Mining (KDD)},
  page 990–998, 2008.

\bibitem{Leiden}
V.A. Traag, L.~Waltman, and N.J. van Eck.
\newblock From louvain to leiden: guaranteeing well-connected communities.
\newblock {\em Scientific Reports}, 9(5233), 2019.

\bibitem{Traud:2011fs}
Amanda~L Traud, Eric~D Kelsic, Peter~J Mucha, and Mason~A Porter.
\newblock Comparing community structure to characteristics in online collegiate
  social networks.
\newblock {\em SIAM Rev.}, 53(3):526--543, 2011.

\bibitem{traud2012social}
Amanda~L Traud, Peter~J Mucha, and Mason~A Porter.
\newblock Social structure of {F}acebook networks.
\newblock {\em Phys. A}, 391(16):4165--4180, Aug 2012.

\bibitem{Ts15}
Charalampos~E. Tsourakakis.
\newblock The k-clique densest subgraph problem.
\newblock In {\em Proceedings of the 24th International Conference on World
  Wide Web}, pages 1122--1132, 2015.

\bibitem{TPM17}
Charalampos~E. Tsourakakis, Jakub Pachocki, and Michael Mitzenmacher.
\newblock Scalable motif-aware graph clustering.
\newblock In {\em Proceedings of the 26th International Conference on World
  Wide Web}, pages 1451--1460, 2017.

\bibitem{WF94}
Stanley Wasserman and Katherine Faust.
\newblock {\em Social Network Analysis: Methods and Applications}.
\newblock Structural Analysis in the Social Sciences. Cambridge University
  Press, 1994.

\bibitem{WaSt98}
D.~Watts and S.~Strogatz.
\newblock Collective dynamics of `small-world' networks.
\newblock {\em Nature}, 393:440--442, 1998.

\bibitem{ground}
Jaewon Yang and Jure Leskovec.
\newblock Defining and evaluating network communities based on ground-truth.
\newblock {\em Proceedings of the ACM SIGKDD Workshop on Mining Data
  Semantics}, 2012.

\end{thebibliography}
\end{document}